\documentclass[a4paper,fleqn]{cas-dc}
\usepackage[numbers,sort&compress]{natbib}
\usepackage{import}
\usepackage{hyperref}
\usepackage{longtable}
\usepackage{amsmath}
\usepackage{amsfonts}
\usepackage{enumerate}
\usepackage{epstopdf}
\usepackage{subcaption}
\usepackage{algpseudocode}
\usepackage{lscape}
\usepackage{amssymb}
\usepackage[linesnumbered,ruled,vlined]{algorithm2e}
\usepackage{url}
\usepackage{amsthm}
\usepackage{array,graphicx}
\usepackage{setspace}
\usepackage{geometry}
\usepackage{multirow}
\usepackage{appendix}
\usepackage{mathtools}
\usepackage{wrapfig}
\usepackage{sidecap}
\usepackage{textcomp}
\usepackage{calc}
\usepackage{pgfplots}
\pgfplotsset{compat=1.16}
\usepgfplotslibrary{statistics}
    
\usepackage{colortbl}
\usepackage{hhline}
\usepackage{makecell}
\usepackage{rotating}
\usepackage[utf8]{inputenc}
\usepackage[english]{babel}
\newtheorem{assumption}{Assumption}
\newtheorem{theorem}{Theorem}
\newtheorem{corollary}{Corollary}
\newtheorem{lemma}{Lemma}
\theoremstyle{definition}
\newtheorem{definition}{Definition}
\usepackage{supertabular,booktabs}

\def\tsc#1{\csdef{#1}{\textsc{\lowercase{#1}}\xspace}}
\tsc{WGM}
\tsc{QE}
\tsc{EP}
\tsc{PMS}
\tsc{BEC}
\tsc{DE}
\begin{document}
\let\WriteBookmarks\relax
\def\floatpagepagefraction{1}
\def\textpagefraction{.001}
\shorttitle{On Feasibility of Defect Prediction Models in Real-World Testing Environments}
%\shortauthors{Umamaheswara Sharma and Ravichandra}

\title [mode = title]{The Probabilistic Bounds on the Feasibility of the Defect Prediction Models in Real-World Testing Environments}
\author[1]{Umamaheswara Sharma B}
\cormark[1]
\fnmark[1]
\ead[url]{uma.phd@student.nitw.ac.in}

\author[1]{Ravichandra Sadam}
\fnmark[2]
\ead[url]{ravic@nitw.ac.in}
\address[1]{Department of Computer Science and Engineering, National Institute of Technology, Warangal}

\cortext[cor1]{Corresponding author}
\fntext[fn1]{Research Scholar, Department of Computer Science and Engineering, National Institute of Technology, Warangal, Telangana, India, 506001}

\fntext[fn2]{Associate Professor, Department of Computer Science and Engineering, National Institute of Technology, Warangal, Telangana, India, 506001}

\begin{abstract}
The research on developing software defect prediction (SDP) models is targeted at reducing the workload on the tester and, thereby, the time spent on the targeted module. However, while a considerable amount of research has been done on developing prediction models or attempting to mitigate the related issues in developing prediction models, it is still unknown whether the developed prediction model really works in real-world testing environments or not. With this article, we bridge this research gap of finding the feasibility of the developed binary defect prediction model in the real-world testing environments. Because machine learning (ML) applications span over many domains, we hope this article may provide sufficient ground to do research on analysing the feasibility of developed prediction models in the related applications in real-time scenarios.
\end{abstract}

\maketitle

\begin{table*}[ht]
\caption{Utilised symbols and its corresponding explanations}
\centering
\begin{tabular}{l|l}
\hline
\textbf{Notation} &  \textbf{Explanation}\\ \hline
	$P$ & A software project\\
%    $\epsilon$ & An error\\
    $p$ & A probability value\\
	$M_i$ & A software module in the project $P$\\
    $S_n$ & The set of $n$ newly developed software modules in the project, $P$\\
	%$\mathcal{D}_{M_i}$ & Number of defects reported in the module $M_i$\\
	$l$ & The number of newly developed modules that are predicted to be clean\\
	\textit{FOR} & The False Omission Rate\\
	\textbf{Pr}[a] & A probability distribution over a discrete variable\\
	$X_i$ & An indicator random variable to measure the hazard value\\
	$X$ & A random variable comprising a sum of $l$ independent random variables, $X_i, |i| = l.$\\
    $\hat{z}$ & A hazard value taken from the random variable, $X$. Or $\hat{z}=X$\\
    $\mathbb{E}[*]$ & An expected value\\
	$\mu_{\hat{z}}$ & Expectation of the random variable $X$, used to represent the expected hazard rate\\
    $t$ & A time unit\\
    $z(t)$ & A general hazard rate of the system\\
    $\Delta t$ & A change in the time unit\\
    $F(t)$ & The cumulative failure distribution function\\
    $R(t)$ & The reliability function\\
    $f(t)$ & The failure density function or the probability distribution function\\
    $T$ & The random variable used to measure the probability of failure\\
    $\lambda$ & The arbitrary best constant hazard value of the manually tested software\\
    $K,\hat{K}$ & The scaling parameters of the Weibull distribution\\
    $m, \hat{m}$ & The exponent parameter value in the Weibull distribution\\
    $\delta$ & The band length or the interval measured from the expectation\\
    $e$ & An exponent function\\
    $\hat{R}(t)$ & The reliability of the SDP-based software that is calculated from $\hat{z}$\\
    $\mu_{\hat{R}}$ & The expected reliability of the SDP-based software that is calculated from $\hat{z}$\\
	$Y_i$ & An indicator random variable to measure the hazard value\\
	$Y$ & A random variable comprising a sum of $l$ independent random variables, $X_i, |i| = l.$\\
    $\hat{z}_Y$ & A hazard value taken from the random variable, $Y$. Or $\hat{z}_Y=Y$\\
	$\mu_{\hat{z}_Y}$ & Expectation of the random variable $Y$\\
	$\hat{R}_Y(t)$ & The reliability of a SDP-based software that is calculated from $\hat{z}_Y$\\
    $\mu_{\hat{R}_Y}$ & The expected reliability of the SDP-based software that is calculated from $\hat{z}_Y$\\     
    \hline
\end{tabular}
\end{table*}
\section{Introduction}
\label{Introduction}
Software Defect Prediction (SDP) is an active research area, attracted by many research practitioners over many years \cite{compton1990prediction, khoshgoftaar1990predicting, bhutamapuram2021project,Basili1996, Fenton1999, Challagulla2005, Gao2007, Elish2008, Catal2009a, Erturk2016, jin2021cross, Sharma2022Measures}. The SDP models are being built to make it easier for the software tester to work on the problematic code instead of searching for the defective code in the newly developed software project. These models are built based on machine learning (ML) techniques that use the defect data collected from the developed project(s). Typically, the defect data is collected for the modules\footnote{The literature on SDP treats any software program, function, class, or method as a module \cite{bhutamapuram2021project, Challagulla2005, Elish2008, zimmermann2009cross, Sharma2022Measures}.} that are present in the previously developed software project. The defect data for the software modules is represented using the well known metrics suits such as size metrics (such as lines of code (LoC)), complexity metrics \cite{mccabe1976complexity}, Halstead's volume metrics \cite{halstead1977elements}, etc. These metrics are treated as independent features to train the machine learning model. Later, using the trained SDP model, the defect-proneness of the newly developed software module is determined.

Several types of SDP have been introduced as a result of research advances, including with-in-project defect prediction (WPDP) \cite{bhutamapuram2021project}, cross-project defect prediction (CPDP) \cite{zimmermann2009cross, Sharma2022Measures} and heterogeneous defect prediction (HDP) \cite{nam2017heterogeneous}. The WPDP models are built for a single software, where the data of the previously developed versions of the software is used to train the model and the module information from the newly developed version of the same software is treated as the test data \cite{bhutamapuram2021project}. The models for CPDP are built on defect data, which is collected from multiple source projects. Here, the class labels for the modules in a currently developed project (also called the target project) are observed by utilising such trained CPDP models. The availability of the common metric space among all the source and target projects is the only criteria for training the CPDP model \cite{herbold2017comparative, Sharma2022Measures}. Whereas the HDP models provide solution when the common metric space between the source and target projects is unavailable \cite{nam2017heterogeneous}. However, each type of the SDP model has the common task of finding the defect-proneness of a newly developed software module.

Once a better prediction model is developed, any organisation may utilise its services. In general, there are chances that the developed prediction model may produce misclassified predictions on the target project, as any state-of-the-art prediction model classifies an unseen example into the respective class with a minimal error \cite{james2013introduction, shalev2014understanding}. Here, the misclassifications are either \textit{false positives} or \textit{false negatives} or both. The following few examples illustrate the plausible actions that will be taken by the testers, given the misclassified predictions. Assume that the newly developed module is clean, and the prediction model classifies it as defective. The tester then observes the \textit{false positive} result. In this case, the tester will check for the correctness of the module. On the contrary, if we assume the developed module is defective and the prediction model classifies it as clean, then the tester observes the \textit{false negative} result. In this case, the tester will miss the defective module. These unrecognised defects in such modules are also called as \textit{soft defects} or \textit{dormant defects}, and in general, they are discovered at the operational phase by triggering such defective modules; otherwise, they will remain dormant in the system \cite{lyu1996handbook}. These dormant defects are harmful to the software because, if they become active, they may lead to the failure of the entire system \cite{lyu1996handbook}.

In manual testing procedures, failures occur in software owing to either design defects or malicious logic, etc. \cite{lyu1996handbook, pressman2005software}. However, in a general software development life cycle (SDLC) with well-established testing procedures, the testers assess the developed modules to remove or modify the defective content \cite{pressman2005software}. After several iterations of testing, over a period of time, the quality of the system will be improved before delivering the product \cite{lyu1996handbook, pressman2005software}. But, in the case of using SDP models, the testers depend only on the predictions, and if the SDP outputs any false negatives, then, as discussed above, the defective modules will become dormant in the system (assuming integration testing, system testing, or acceptance testing does not prompt these defects). In such a case, achieving the quality of the software system is difficult.

To understand the amount of impact that the dormant defects have on the developed software system, in this work we provide a theoretical analysis based on the probabilistic tail bounds. In other words, in this work, through the probabilistic bounds, we provide insights on the feasibility of the developed binary defect prediction models in the real-world testing environments. The bounds we provide in this work are largely dependent on some assumptions (a detailed discussion of the stated assumptions is given in Sections \ref{Preliminaries} and \ref{Discussion}). The preliminary assumption is that the actual defect-proneness of the developed module is unknown. Later, we model the outcome of each predicted clean module as an independent random variable. At the time of developing the prediction model, we assumed similar distributions among the training, testing, and population data. As a result, we assign each misclassified defective module (or random variable) a similar probability value. Here, we use the value of \textit{False Omission Rate (FOR)} (which is observed on the test dataset) as the probability of occurrence of a single failure in the software as a result of the misclassification of a defective module. Since each random variable defines a single failure condition in the developed software, the sum of the random variables (which will be defined for the predicted clean modules) represents the total failures in that software (at any time, $t>0$). We treat the resultant failure instances in the software as ``hazard instances". Now, we observe the hazard rate of the same software, but it has been completely tested by a group of testers.

Using the Chernoff bound technique \cite{hellman1970probability, motwani1995randomized}, we provide bounds on the occurrence of fewer hazards in the software that uses SDP than the total hazards in the same software (at any point in time, $t>0$) that is tested by humans. Similarly, using the same technique, in terms of the reliability of the various hazard models, we provide bounds on the occurrence of greater reliability in the software that uses SDP than the reliability of the same software (at any point in time, $t>0$) that is tested by humans. We assume a manually-tested software follows the Weibull distribution of the \textit{hazard rate} (and its other supplements), in deriving various proofs. Assume, for example, that manually tested software follows a Weibull distribution of the hazard form ($Kt^m$). Then we derive a maximum bound of $e^{\frac{-(lp-Kt^m)^2}{2lp}}$ as the probability of occurrence of fewer hazards in the software (that uses SDP) than the total hazards in the same but manually tested software. Similarly, in terms of the reliability of a software that is modeled on the Weibull distribution of the hazard rate, we derive a maximum bound of $e^{-\big[e^{lp(e^{-t}-1)}-\frac{Kt^m}{m+1}\big]^2\frac{1}{2e^{lp(e^{-t}-1)}}}$ as the probability of occurrence of more reliability in the software (that uses SDP) than the reliability of the same but manually tested software.

In addition to that, we derive more bounds (in terms of both the hazard rate and reliability) by assuming each random variable (which is defined for the misclassified defective module) generates a Weibull distribution of hazard instances instead of a single failure in the system.

%The proofs (which are provided in Sections \ref{The Proofs}, \ref{Supplements}, and \ref{Discussion}) are majorly dependent on the parameters such as $K, \hat{K}, m, \hat{m}$, and the time $t$ (the details about these parameters are provided in Sections \ref{The Proofs}, \ref{Supplements}, and \ref{Discussion}). 
The derived bounds will become more sharp (or even these will become close to zero) when we observe more deviation between the random variable (assume it is represented for hazards in a software that is tested using an SDP model) and the real hazard rate of the same but manually tested software. In such a case, we conclude the ineffectiveness of the developed prediction model.
\subsection{Contribution} 
    This work is targeted at showing the feasibility of the SDP models in real-world testing environments. To the best of our knowledge, this work makes the following unique yet important contributions in the field of software defect prediction:
\begin{enumerate}
    \item We have provided tight lower bounds on the feasibility of the developed defect prediction model in real-world testing environments. The bounds are discussed in terms of both the hazard rate and the reliability of the software.
\end{enumerate}
To the best of our knowledge, this is the first work that provides a critique of the developed binary-prediction model (in our case, we used the software defect prediction model as a tool to provide the analysis) in real-world testing environments.\\
\textit{Paper Organisation}: In Section \ref{Preliminaries}, we provide some assumptions to derive the bounds. In addition, we determine the probability (from the false negative injections) of a failure of the software and its distribution. Section \ref{Hazardrate and Reliability} provides the information on the reliability indices such as hazard rate, reliability. In addition, in this section, we describe the revised bath-tub curve to demonstrate the software reliability. The proofs using the Chernoff bound are provided in Section \ref{The Proofs}. The proofs from the supplements of the Weibull distribution on the reliability indices are discussed in Section \ref{Supplements}. In Section \ref{Discussion} we provide a detailed analysis of the derived proofs. Section \ref{Conclusion} concludes the work and provides a potential future dimensions.
\section{Preliminaries}
\label{Preliminaries}
\begin{table*}
\centering
\caption{Confusion Table.}
\label{ConfusionTable}
\begin{tabular}{lc|cc|}
\cline{3-4}
 & \multicolumn{1}{l|}{} & \multicolumn{2}{c|}{\textbf{Predicted condition}} \\ \cline{2-4} 
\multicolumn{1}{l|}{} & \textbf{\begin{tabular}[c]{@{}c@{}}Total examples \\ (Defective + Clean)\end{tabular}} & \multicolumn{1}{c|}{\textbf{Defective}} & \multicolumn{1}{c|}{\textbf{Clean}} \\ \hline
\multicolumn{1}{|c|}{\multirow{2}{*}{\textbf{Actual condition}}} & \textbf{Defective} & \multicolumn{1}{c|}{TP (Hit)} & FN (Type-II error or miss) \\ \cline{2-4} 
\multicolumn{1}{|c|}{} & \textbf{Clean} & \multicolumn{1}{c|}{FP (Type-I error, false alarm)} & TN (Correct rejection) \\ \hline
\end{tabular}
\end{table*}
In Section \ref{Introduction}, we discussed various types of SDP, such as WPDP, CPDP, and HDP. However, these models are built for the binary classification (either \textit{defective} or \textit{clean}) of the newly developed software modules. Our objective is independent of the type of the SDP model. That is, all the proofs provided in Sections \ref{The Proofs}, \ref{Supplements}, and \ref{Discussion}, are applicable to any type of binary classification of defect prediction models because they each have the common goal of finding the defect proneness of the newly developed software module.

Now, defining the probability is the first step in providing the tightest bounds for the feasibility of the SDP models. In Section \ref{Deriving the Probability}, a detailed explanation is provided for determining the probability from the predictions (particularly using the confusion matrix). From the foundation step in Section \ref{Deriving the Probability}, in Section \ref{TheDistribution}, we define the binomial distribution that is derived from the sum of numerous Bernoulli trials. Section \ref{Real-time Applicability} provides details about the real-time applicability of the SDP models.
\subsection{Deriving the Probability}
\label{Deriving the Probability}
To provide tight bounds on the maximum possible reliability and minimum possible hazards in a software which uses the SDP models when compared with the manually tested software, we begin by counting the chances of failures in the software system, from the predictions of the SDP model. Of which, the major possible chance of failure is when the defective module is predicted as \textit{clean}. If the prediction model assigns a clean label to the defective module (assuming the newly developed software module contains defects in it), then the tester may miss such a module. Hence, there is a chance that the software fails.

Before modelling the failures in a software, in this section, we derive a probability value as the percentage of misclassification of defective modules over the total observed cleans. Before defining the probability and its other dependent derivations, we make the following assumptions:

Now, the following assumption sets the binary outcome for every newly developed software module:

\begin{assumption}
\label{Assumption1}
The SDP models operate on a binary classification; hence, the predictions for the software modules are either defective or clean (but not both).
\end{assumption}
That is, irrespective of the number of defects in a defective module, the prediction model will assign a binary value to such a software module. With this assumption, we infer that, even if the software module contains more than one defect, a single cause may lead to the failure of that software module.

\begin{assumption}
\label{Assumption2}
In real-world testing environments, an organisation employs the traditional SDP model, which is built on any batch learning approach.
\end{assumption}
In general, the prediction models are built based on either the fixed set of training data or the variable set (dynamic data) of training data. If the prediction model is trained using a fixed set of training data (in this case, the training data is drawn from the history of the project data), then we call the training process \textit{batch learning}. On the contrary, if the prediction model utilises a dynamic set of training data to train the model, then we call the learning process \textit{online learning}. In general, in the literature, the SDP models utilise the historical data of the released software projects in the training \cite{bhutamapuram2021project, zimmermann2009cross, nam2017heterogeneous}.

\begin{assumption}
\label{Assumption3}
At the time of developing the prediction model, the data distributions for the training set, test set, and population were similar.
\end{assumption}
This assumption will ensure that the prediction model may produce near-similar results on all the datasets. Without loss of generality, the statistical validity of the developed prediction models is provided based on this assumption. That is, in general, the SDP (WPDP, CPDP, or HDP) models work on the basis of similar source data distributions with the target project data. Consequently, these distributions reflect similar distributions with the population data \cite{Sharma2022Measures, jin2021cross, nam2017heterogeneous}.

In general, we observe from the confusion table \ref{ConfusionTable} that the predicted clean modules belong to either false negatives or true negatives. Now, the only possible instance where the software fails is when the defective module is predicted to be the clean module. Now, the following assumption ensures that there is a single failure for each misclassified defective module.

\begin{assumption}
\label{Assumption4}
Each defective module misclassification can result in a single failure in the software system.
\end{assumption}
From the predicted clean modules, from Assumption \ref{Assumption1}, the above assumption ensures that each misclassified defective instance (we call them ``dormant defects'') can cause a single failure in the software system. In general, in some cases, a single defect may cause multiple failures, and these failures are assumed to be in a Poisson process \cite{ross2014introduction}. But for a simple case, we assume that each defect may cause a single failure in the software. Hence, as a result, the total number of dormant defects (or hidden failures) in the software is equivalent to the total number of false negative instances.

Now, to measure the percentage of misclassified defective instances (later, this is used to measure the percentage of failure instances in newly developed software) from the predicted clean modules (that is, from the test dataset), we use an evaluation metric called \textit{False Omission Rate (FOR)}. The \textit{FOR} measures the percentage of the \textit{type}-II errors from the observed clean modules. From the confusion table \ref{ConfusionTable}, the \textit{FOR} is defined as the ratio of false negatives over the total number of negative calls. This is given below:
\begin{equation}
\label{FOR}
    \textit{FOR} = \frac{\text{False Negatives}}{\text{False Negatives+True Negatives}} = \frac{\text{False Negatives}}{\text{Negative Calls}}
\end{equation}

%The value of the measure \textit{FOR} varies according to the nature of the prediction model. That is, if the machine learning model utilises a fixed set of training instances then prediction model follows the batch learning. Note that, here, the test data is used to determine the performance of the developed prediction model. Similarly, if the machine learning model utilises a variable set of training instances (online machine learning), then the performance of the model on the variable set of test data will change at regular intervals. If an organisation utilises the services of the traditional SDP model (developed based on batch learning), then the value of the measure \textit{FOR} is treated as a constant value. The proofs provided in this work are based on the assumption that the SDP model utilises any batch learning approach.

Equation \ref{FOR} ensures that the percentage of failure instances in the target software project\footnote{In general, the SDP models are trained on the historical source project data (we use this augmented data to train the prediction model) and validated on the target project data.}.

Now, let us define the probability that the newly developed defective module falls into the clean class. Since from Assumption \ref{Assumption2}, we assume that an organisation utilises the batch learning model to predict the defect proneness of the newly developed software module (this assumption ensures some arbitrary constant FOR value), and also from Assumption \ref{Assumption3} we have that the distributions of the test set and the population set are similar (this assumption ensures near uniform predictions from the prediction model) then, in real-time scenario, any predicted clean module will become a wrong prediction with the probability of FOR. This is given as: 
\begin{equation}
\label{FOR-Probability}
    p = \textit{FOR} = \frac{FN}{FN+TN}
\end{equation}

Equation \ref{FOR-Probability} is valid in the real-time testing environments because, from Assumption \ref{Assumption2}, we know the fact that the testers use the developed prediction models (that is, these prediction models use the historical data in the training), and for every newly developed module, if these are predicted as clean, then each predicted clean module has a similar probability $p$ to be an originally defective module. Precisely, in real-time testing of the software modules, if these modules are tested and predicted as clean by the SDP model, then all these are assumed to be the wrong predictions, with each having a similar probability, $p$.

Now, using the probability, $p$ (from Equation \ref{FOR-Probability}), let us define the distribution for the set of predicted clean modules.
\subsection{Defining the Distribution}
\label{TheDistribution}
The assumptions \ref{Assumption1}-\ref{Assumption4} are essential in defining the probability distribution. By the end of this section, we will have the total number of failure incidents and the expected number of failure incidents in a software system, which are observed from the SDP model. Since we do not know the actual class labels for the newly developed software modules, we treat each prediction outcome of the newly developed software module as a random variable. Out of all the random variables, we are interested in the random variables that measure the misclassified defective modules in the software project. The following scenario formally depicts the failures in the software that uses the SDP model.

Assume $n$ new software modules are developed in the project, $P$, and each newly developed software module is independently tested with the SDP model to predict its class label. Let the set $S_n = \{M_1, M_2, \dots, M_n\}$ represent the set of newly developed $n$ software modules in the project, $P$ (whose actual class labels are unknown) utilise the SDP models to predict its class labels. The newly developed $n$ software modules are assumed to have been originally either defective or clean (but not both). Note that defective modules may contain more than one defect. But for a simple case, from Assumption \ref{Assumption1}, the defect count will be converted into a binary value. Now, using the SDP model, we observe the predictions for the modules in set $S_n$. Let us assume that $l$ software modules are predicted to be from the clean class out of $n$ software modules. Note that the set of predicted $l$ clean modules was originally either defective or clean. Now, to count the number of wrongly predicted defective modules from the $l$ modules, let us define an indicator random variable that takes 1 if the newly developed module $M_i, |i|=l$, is wrongly classified into the clean class and takes 0 if the newly developed module $M_i$ is correctly classified into the clean class. This is given as:
\begin{equation}
\label{IndicatorRV}
X_i = \begin{dcases*}
        1, & if the module $M_i$ is wrongly classified into clean\\
        0, & if the module $M_i$ is correctly classified into clean
      \end{dcases*}
\end{equation}

Here, we define the indicator random variable $X_i, i \in \{1,2, \cdots, l\}$ to count the number of failure instances in a software that uses the defect prediction model. %According to Assumption \ref{Assumption4}, we assume an occurrence of a single failure from each misclassified defective module. This makes an easy computation in counting the number of failure incidents in the software (that uses SDP model).

From the Assumption \ref{Assumption2}, since we set the fact that an organisation utilises the services of the traditional SDP model (that is developed using batch learning), and from Assumption \ref{Assumption3}, we ensure a similar probability for the predicted value of each newly developed module, then each predicted clean module going into the wrong class is assigned with the probability, $p$. This is represented as:
\begin{equation}
\label{Probability-X}
    \textbf{Pr}[X_i=1] = p, \text{ and, }  \textbf{Pr}[X_i=0]=1-p, \text{ for } 1 \leq i \leq l.
\end{equation}

where $X_i$ is an indicator random variable defined in Equation \ref{IndicatorRV}. Since each indicator random variable $X_i, i \in \{1,2, \cdots, l\}$ defined in Equation \ref{IndicatorRV} takes the value 1 with the probability $p$ (similarly, $X_i$ takes 0 with the probability 1-$p$), this becomes a Bernoulli trial. Note that the value $p$ = 0 indicates that the prediction result does not have the \textit{false negative} instances (on the test set), and hence, the SDP model is working accurately towards the defective modules. On the contrary, the value $p$ = 1 indicates that the prediction result does have the \textit{true negative} instances and, hence, the prediction model is not working accurately for the actual clean modules. Thus, all the predicted clean modules belong to the \textit{false negative} category. Now, the following assumption is essential to providing all the proofs (which are given in Sections \ref{The Proofs}, \ref{Supplements}, and \ref{Discussion}):

\begin{assumption}
\label{Assumption5}
The SDP model should exhibit at least one false negative and one true negative outcome for the target project.
\end{assumption}

This assumption ensures that the probability ($p$) value lies in the interval (0, 1). In Sections \ref{The Proofs}, \ref{Supplements}, and \ref{Discussion}, we provide all the proofs based on this assumption. However, in supporting this assumption, it is a fact that, according to Herbold et al. \cite{herbold2017comparative}, most SDP models do not work well on the target datasets. In addition, finding the best classifier that works consistently for the specific application remains a difficult task \cite{bishop2006pattern}. Hence, we believe this assumption holds true in real-world testing environments.

Now, using the group of Bernoulli random variables,  $X_i, i \in \{1,2, \cdots, l\}$, below we define the binomial distribution.
\subsubsection{The Binomial Distribution for the Failure Incidents:}
Because the SDP model independently and randomly takes each developed software module as input to find its defect proneness, each $X_i, i \in \{1,2, \cdots, l\}$, becomes a Bernoulli random variable. Our discussion below will focus on the random variable $X$, which is defined using the sum of the independent Bernoulli random variables with identical probabilities. Hence, the random variable $X$ has a binomial distribution. This is expressed as:
%Now, let a random variable $X$ that follows a Binomial distribution (with the parameters $l$ and $p$ is denoted as $X\sim B(l,p)$) is used to count the total number of failure (hazard) incidents in the software. This is represented as:
\begin{equation}
\label{SumofIndependentTrials}
    X = \sum_{i=1}^{l} X_i
\end{equation}
The random variable derived in Equation \ref{SumofIndependentTrials} is nothing but the hazard rate (at any time $t>0$, after deploying the software) in SDP-based software. Let $\hat{z}(t)$ is the hazard rate of the software that is tested using the SDP model, and then we have that:
\begin{equation}
\label{Hazard Rate-SDP-X}
    \hat{z}(t) = X
\end{equation}
Since, the Bernoulli random variables $X_1, X_2, \cdots, X_l$ are assumed to be independent then, as mentioned above, the sum of $l$ identical Bernoulli trials is said to be a binomial distribution. The mean of the binomial distribution or the expected number of hazard instances in the software is derived as follows (by using the linearity of expectation):
\begin{multline}
\label{Expectation}
    \mu_{\hat{z}} = \mathbb{E}[X] = E\Big[\sum_{i=1}^{l} X_i\Big] = \sum_{i=1}^{l} \mathbb{E}[X_i] =\\ \sum_{i=1}^{l} \big[ 1*\textbf{Pr}[X_i=1] + 0*\textbf{Pr}[X_i=0]\big] = \sum_{i=1}^{l} p = lp
\end{multline}

Equation \ref{Expectation} provides information about the expected number of failure (hazard) incidents in the software that is tested by using the SDP models, assuming each wrongly predicted defective module commits to a single failure (according to Assumption \ref{Assumption4}). Note that, throughout the paper, we use the terms failure rate and hazard rate interchangeably.

Note that, from Equation \ref{Hazard Rate-SDP-X}, in the real-time scenario, after the software is deployed (or released) to the clients, it may have $X$ (or, on average $\mu_{\hat{z}}$) defects (consequently, each defective instance generates a single failure) in the system. Now, we provide the tight bounds by comparing the random variable $X$ (which is a measure of the hazards in a software that is tested using SDP models) with the hazard rate of the same software that is tested by the group of testers. We provide the proofs using the tail-inequality technique known as the Chernoff bound. The expectation of the binomial distribution (from Equation \ref{Expectation}) is essential and sufficient to provide proofs using the Chernoff bound. A detailed list of proofs are provided in Sections \ref{The Proofs}, \ref{Supplements}, and \ref{Discussion}.

%The obtained tight probability bound value indicates that the SDP models are either feasible or impractical in the real-world testing environments. For example, larger probabilistic bound value indicates the successful working nature of the SDP models in the real-time testing environments. On the contrary, smaller probabilistic bound value indicates the impracticability of the SDP models in the real-world testing environments. The proofs and the detailed discussion on the proofs are provided in Sections \ref{The Proofs}, \ref{Supplements}, and \ref{Discussion}.
\subsection{Real-time Applicability}
\label{Real-time Applicability}
Assume a new software (with $n$ modules) is developed by an organisation with a fixed set of specifications intended to fulfil the collected requirements. After developing the software (assuming the developers follow any specific procedure to develop software), a group of testers is put on the developed software to eliminate the defects present in that system. The common strategy is to test the individual modules with a set of test cases. Assume that any module in the software is intended to perform according to the well-written specifications, and these specifications will become the test-cases for the developed software module \cite{lyu1996handbook}. Now, within the available amount of time (since any testing team must have time limits to accomplish testing the software), each module will be carefully tested by the group of testers to observe (and then remove) the defects. After the thorough revisions of the software, most defects that cause failure of the software will be removed from the system. However, it is evident that most software systems still fail due to many reasons, ranging from improper testing practises to unexpected external factors.

As discussed in Section \ref{Introduction}, instead of manual testing, the SDP models try to simplify the testing job and provide the defect-proneness of the developed software modules. Based on the predictions, testers conduct a code walk on the predicted defective modules to remove any damaged code. But, aforementioned, if the defective modules are predicted as clean then there is a chance that the software may fail. 

In the above two scenarios, the software systems may experience failures. Now, before finding the amount of deviation between the two quantities of failures (one is observed on the manually tested software and the other is a random variable), the following assumptions must hold true.

%In all the constructed proofs (presented in Sections \ref{The Proofs}, \ref{Supplements}, and \ref{Discussion}), we compare the hazard rate and reliability of the two similar software systems. Now, the following assumption ensures similar software for two kinds of testing procedures (that is, manual testing and testing using SDP models).
\begin{assumption}
\label{Assumption6}
The software is similar for both manual testing and testing using SDP models.
\end{assumption}
%In Sections \ref{The Proofs} and \ref{Supplements}, all the proofs for the feasibility of SDP models are derived based on this assumption. 
This is a primary and important assumption in conducting this research. Since we are interested in finding the feasibility of the SDP models, we provide the proofs for the deviation of the hazard rate of the software that is tested using the SDP models below from the hazard rate of the same but manually tested software. Also, we provide the proofs by using the deviation of the reliability of the software that is tested using the SDP models above from the reliability of the same but manually tested software.

Note that, according to Assumption \ref{Assumption6}, even though the software is similar in both the cases of manual testing and testing using SDP models, the number of modules and the failure incidents may differ. This is because manual testing follows a typical structural methodology to eliminate dead code, unreachable code, or some redundant code from newly developed software. As a result, after thorough testing, the majority of such code will be eliminated while keeping the original functionality intact, before releasing the product. This may result in a variation in the total number of modules (when compared with the initial number of developed modules) before releasing the product. Similarly, when testing the software using SDP models, the testers also test the predicted defective modules. For such modules, after thorough testing, the majority of unwanted code will be removed. As a result, the total number of modules in the same software may vary. Consequently, we may see differences in the total hazards between the two software systems, one of which is tested using the SDP model and the other by a group of testers.

Now, the following assumption ensures the presence of hazards in software that is tested using SDP models:

\begin{assumption}
\label{Assumption7}
For the incorrectly predicted defective modules, the integration test, system test, or acceptance test does not prompt the defects.
\end{assumption}
This is a rare but possible situation in real-world testing environments. Because testing the software (via manual testing) with an exhaustive set of test cases is impossible, the presence of hidden defects in the software system is unavoidable. Hence, end-users may experience failures after triggering such hidden defects. Similarly, for software that uses SDP models, the presence of defects is inevitable even after the completion of testing \cite{bhutamapuram2021project, lyu1996handbook}. Specifically, while the SDP models reduce the work load on the testers by avoiding unit testing, the testers may encounter defects during integration testing, system testing, or acceptance testing. However, according to Assumption \ref{Assumption7}, we will remain with the $X$ hazards (or, on an average, $\mu_{\hat{z}}$ hazards) in the software system after the software is released. 

In addition to the above assumption, the following assumption ensures the absence of hazards from the serviced defective modules:

\begin{assumption}
\label{Assumption8}
The repaired defective modules will not deteriorate.
\end{assumption}
With this assumption, we can only generate hazardous conditions from the predicted clean modules. Relaxing this assumption may result in a prediction of the likelihood of hazard occurrences from the serviced code. % (after prediction, the tester assists only the defective modules), and which is similar to estimating the hazard rates from the manually tested software.

Now, the following assumption ensures similar time intervals for estimating the hazards and reliability of both the software systems (that is, the software that is tested using SDP and the other that is manually tested).
\begin{assumption}
\label{Assumption9}
Both the software systems (that is, the software that is tested using SDP and the other, manually tested software) will be released at time 0, and their hazards and reliability are measured in the time interval [0, $t$].
\end{assumption}
Now, in order to validate the situation where the binomial distribution is applicable in the real-time scenario, according to Assumption \ref{Assumption9}, all the newly developed modules are tested using the SDP model before time 0. That is, because the SDP model predicts whether the newly developed module will be defective or clean, we assume that the testers will service the predicted defective modules before releasing the product. Later, for the software system that uses the SDP models, the hazard rate and its reliability are measured in the time interval [0, $t$]. This enables easy computations for the reliability indices on the software that uses the SDP model. Similarly, according to Assumption \ref{Assumption9}, for a manually tested software, the testers conduct a complete system test before releasing the product. After a thorough system test, the software will be deployed at time 0, and thereafter the hazard rate and its reliability will be measured in the time interval [0, $t$]. Similar to testing the software using SDP models, this assumption enables easy computations for the reliability indices in the case of manually tested software.
\section{The Hazard Rate and Reliability of Software}
\label{Hazardrate and Reliability}
In the previous section, we estimated the expected number of hazards that may occur in the software that uses the SDP model, assuming the software is deployed. Later, these hazards (which is a random variable) are used to define the reliability of the software system. In order to calculate the reliability from the hazard rate, in Section \ref{Software Reliability Indices} we provide the definition and the relation between such reliability indices. The definition of reliability indices is common for both manually tested software and software that is tested using the SDP models.

In general, the reliability of a hardware system is represented in the bath-tub curve \cite{lyu1996handbook}. The software reliability, however, does not show the same characteristics as the hardware \cite{hartz1997introduction}. Hence, in Section \ref{The BathTub Curve}, we discuss the details of the revised bath-tub curve for the software reliability. In Section \ref{The Weibull Distribution}, we discuss a widely used hazard model called the Weibull distribution. Also we derive the reliability of a software from the the Weibull distribution of the hazard model.

%To estimate the probability that the SDP model achieves the reliability of the software, in this work, we make a comparison with the reliability of the software which involves manual testing. In practice, there are various methods available to estimate the reliability of the software.  is established  analyzing the historical failure data. For this various software reliability growth models have been proposed in the literature \cite{aggarwal2022optimization} \cite{huang2022software}.%For this, we first provide information about the hazard rate and the reliability of the software system. Also, we compute the hazard rate and the reliability of the software system which uses the SDP model. Later, using Chernoff bound, we calculate the maximum and minimum possibility of the reliability of the software which uses SDP model in the real-time scenario.%Since, we are considering the false negative instances (which are later assumed to be the failure cases in the software system) to define the random variable \textit{X}, we find the tight bounds that the random variable \textit{X} exceeding the number of failure instances in a ideal software system. The number of failure instances in a ideal software system is observed from the reliability index called hazard rate $z(t)$. Section \ref{Software Reliability Indices} provides information about the reliability indices, which are used to define the number of failure instances in a ideal software system. Without loss of generality, the number of failure instances are then used to define the reliability of the software.
\subsection{Software Reliability Indices}
\label{Software Reliability Indices}
This section defines and explains the relationship between reliability indices such as hazard rate and reliability. We adapted the fundamental definitions of the reliability indices from the work of Lyu in  \cite{lyu1996handbook}.

\subsubsection{Hazard Rate}
The hazard rate is the instantaneous rate of failure of the system at time $t$, given that the system survives up to the time $t$ \cite{lyu1996handbook, hartz1997introduction}. This is expressed as:
\begin{equation}
\label{Eq: Hazard Rate}
    z(t) = \lim_{\Delta t \to 0} \frac{F(t+\Delta t) - F(t)}{\Delta t R(t)} = \frac{f(t)}{R(t)}
\end{equation}
Where, $f(t)$ and $F(t)$ are the probability density function (pdf) (or the failure density function) and cumulative distribution function (cdf), respectively. $R(t)$ is the reliability function, used to measure the probability of success at time $t$. In this work, we assume the availability of the hazard rate of manually tested software. Hence, functions such as $f(t)$ and $F(t)$ are not required to estimate.
\subsubsection{Reliability}
Before defining the reliability, we assume the random variable of interest is the time to failure of the software, $T$. Now, the software reliability is defined as the probability of failure-free software operation for a specified period of time in a specified environment \cite{lyu1996handbook, hartz1997introduction}. Formally, this is expressed as:
\begin{equation}
\label{The Reliability Definition}
    R(t) = Pr[T>t] = 1 - F(t) = \int_t^\infty f(x) dx
\end{equation}
\subsubsection{Relation between Hazard Rate and Reliability}
After few substitutions and derivations from Equations \ref{Eq: Hazard Rate} and \ref{The Reliability Definition}, we obtain the relation between the hazard rate and the reliability as \cite{lyu1996handbook, hartz1997introduction}:
\begin{equation}
\label{Reliability from Hazard Rate}
    R(t) = e^{-\int_0^t z(x) dx}
\end{equation}
Equation \ref{Reliability from Hazard Rate} helps in deriving the reliability values for the various hazard models. Estimating the reliability from the hazard rate is similar for the software that uses the SDP model and the manually tested software.
\subsection{The Bath-Tub Curve for Software Reliability}
\label{The BathTub Curve}
In general, any software system is often serviced by a group of testers for every occurrence of defects or those deemed to enhance the functionalities of that software. Hence, with every modification in the system by removing the observed defects, over a period of time, the failure incidents in the software will be reduced to their minimal value. As a result, the software does not wear out over a period of time but may experience failures based on improper or misunderstood specifications, input data errors, algorithmic errors, programme logic errors, etc. \cite{lyu1996handbook}. The typical hazard curves that may likely occur in any software, in a given period of time, are represented in Figure \ref{RevisedBathTub}.

\begin{figure}
\centering
\includegraphics[width=9cm]{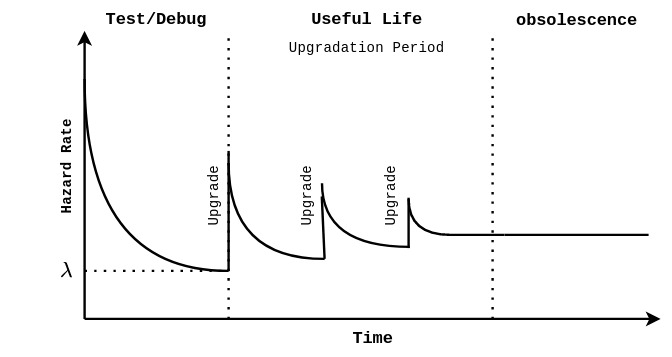}
\caption{A typical revised bath-tub curve for the software reliability \cite{hartz1997introduction}.}
\label{RevisedBathTub}
\end{figure}

The hazard rate modelling in the software is different when compared with the hardware. In \cite{keene1994comparing}, Keene has provided a partial list of the distinct reasons for observing the variation in the hazard rates of the software system. In \cite{hartz1997introduction}, Mary Hartz et al. has presented a typical revised bath-tub curve, and that is observed in Figure \ref{RevisedBathTub}. From Figure \ref{RevisedBathTub} it is observed that, the hazard rate of the software is divided into three regions, such as test/debug, upgradation, and stable/obsolescence region. The region I in Figure \ref{RevisedBathTub} describes that, at the time after deployment, the software may experience more failures upon triggering the defects that are present in the written software programmes (modules). As a result, the software may experience a higher hazard rate in this region. The quality control mechanisms and initial product testing strategies help to eliminate the observed defects in the system. In this region, the hazard rate tends to decrease as a function of time.

The region II in Figure \ref{RevisedBathTub} represents the upgradation period of the software. In this region, the software system experiences more upgrades due to either removing the defective functionalities or enhancing the functionalities in that software \cite{hartz1997introduction}. As a result, the hazard rate fluctuates at every occurrence of a new upgrade in the software system. Sometimes, it is likely to observe cascading defects (consecutively, cascading failures) from the bug fixes in the software system, or sometimes every upgrade in the software will lead to a reduction in the number of failure incidents in the system.

The region III in Figure \ref{RevisedBathTub} represents the stable nature of the software as it does not invite any new upgradations. Or, in this region, the software is often treated as obsolescence, indicating no motivation for providing feature updates in the system. However, if the software is in use, the end-user may experience fewer failure incidents in this time period \cite{lyu1996handbook, hartz1997introduction}. Hence, the hazard rate of the system in this region is treated as an arbitrary constant value as opposed to the hazard rate in the hardware systems \cite{lyu1996handbook}. Since the software becomes more complex with each upgrade, possibly more defects will be eliminated from the system. As a result, the software system will be stable for a longer period of time.

Due to the presence of the variation in the hazard rates in a software over a period of time (from Figure \ref{RevisedBathTub}), we provide the proofs based on assuming the hazard model follows the Weibull distribution. The Weibull distribution provides different functions of the hazard rates over a period of time. The literature exhibits estimating many complex hazard functions and the reliability models for the software \cite{lyu1996handbook, hartz1997introduction, huang2022software, gokhale1996unification}. However, for easy computation of the bounds, we are utilising the Weibull distribution forms of the hazards in software. The Weibull model of the hazard rate in the software is described in Section \ref{The Weibull Distribution}.
\subsection{The Weibull Distribution}
\label{The Weibull Distribution}
In many cases, the hazard rate of the software in a given time period does not follow a specific function form \cite{lyu1996handbook, pressman2005software}. In order to approximate various hazard curves, we use a hazard rate of the form called the Weibull distribution.
\begin{definition}
\label{Definition: Weibull Hazard}
For any software system, the Weibull distribution for the hazards function ($z(t)$) is given as \cite{lyu1996handbook}:
\begin{equation}
\label{Weibull-Hazard}
    z(t) = Kt^m, \text{ for some } K>0, m>-1, \text{ and time } t>0
\end{equation}
\end{definition}
The Weibull distribution is a time dependent function and, by appropriate choice of the two parameters $K$ and $m$ in the Weibull distribution resulting in approximating a wide range of specific hazard rates. As an example, for the fixed values of $m$, a change in the parameter $K$ merely results in a change in the vertical amplitude of the hazard rate. For a simple case, substituting the values of $m$ to 0 and $m$ to 1 in the Weibull distribution, we get the \textit{constant hazard rate} and \textit{linearly-increasing hazard rate}, respectively. In Section \ref{Supplements}, we provide the proofs for various supplements to the Weibull distribution of the hazard rate that are likely to occur in the software. These supplements may enable us to substitute the most relevant hazard form in the constructed proofs.

Now, from the Weibull distribution for the hazard rate (from Definition \ref{Definition: Weibull Hazard}), the following lemma defines the reliability of the software:
\begin{lemma}
\label{Lemma-Reliability-Wibull}
For the Weibull hazard model $z(t) = Kt^m, \text{ for some } K>0, m>-1 \text{ and time } t>0$, its reliability is:
\begin{align}
\label{Weibull-Reliability}
    R(t) = e^{\frac{-Kt^{(m+1)}}{m+1}}
\end{align}
\end{lemma}
\begin{proof}
Since we assume that the hazards in a software follows Weibull distribution, substituting $z(t)$ in $R(t)$ (that is in Equation \ref{Reliability from Hazard Rate}) yields:
\begin{align}
\label{Reliability-Weibull-Inintial}
    R(t) = e^{-\int_0^t Kx^m \text{ } dx}
\end{align}

Now, after simplifying Equation \ref{Reliability-Weibull-Inintial}, we get the desired form of the reliability.
\end{proof}
Here, we use $z(t)$ and $R(t)$ to represent the hazard rate and reliability of the manually tested software. It should be noted that the functions $z(t)$ and $R(t)$ will become some arbitrary real-values for any estimated value of the parameters $K$ and $m$ and at any time $t$. 
\section{The Proofs}
\label{The Proofs}
In Section \ref{Bound on the Hazard Rate}, we provide a proof for the possibility of achieving minimum hazards in a software that uses SDP than the hazards in a manually tested software. A tight lower bound is derived in the proof. Similarly, in Section \ref{Bound on the Reliability}, we provide a proof for the possibility of achieving maximum reliability by a software that uses SDP, when compared with the reliability of a manually tested software. For this, a tight upper bound is derived in the proof. In order to provide a tight upper bound, eventually we ended up showing the lower bound of some other form in the proof.
\subsection{A Tight Lower Bound in terms of the Hazard Rate of a Software}
\label{Bound on the Hazard Rate}
In Section \ref{TheDistribution}, we derived the number of hazard (failure) instances in the software that is tested by using the SDP model using a random variable (that is, $\hat{z}(t) = X$). Now, the following theorem defines the deviation of a random variable, $X$ below the hazard rate of a manually tested software, $Kt^m$ (in fact, far below from the expected hazard rate, $\mu_{\hat{z}}$). 
\begin{theorem}
\label{Theorem-Weibull-Hazard}
Let $X_1, X_2, \dots, X_l$ be the independent Bernoulli trials such that for, $1\leq i \leq l$, $Pr[X_i = 1] = p$, where $0 < p < 1$. Then for X = $\sum_{i=1}^{l} X_i, \mu_{\hat{z}} = \mathbb{E}[X] = \sum_{i=1}^{l} p = lp$, and $\exists$ parameters $K>0, m>-1$, time $t>0$, and for the Weibull hazard rate, $z(t) = Kt^m$:
\begin{align}
    Pr[X < Kt^m] < e^{\frac{-(\mu_{\hat{z}}-Kt^m)^2}{2\mu_{\hat{z}}}}
\end{align}
\end{theorem}
\begin{proof}
We know that for some $0<\delta \leq 1$, and $\mu$,  using Chernoff bound, the lower tail bound for the sum of independent Bernoulli trials, $X$, that deviates far from the expectation $\mu$ is \cite{chernoff1952measure}:
\begin{align}
\label{LowerChernoff}
    Pr[X < (1-\delta)\mu] < e^{\frac{-\mu\delta^2}{2}}
\end{align}

Here, the value $(1-\delta)\mu$ represents the deviated value from the expectation $\mu$, with the band length of $\delta$.\\ %, this is depicted in Figure \ref{Image:LowerBound}. 
\iffalse
\begin{figure}
\centering
\includegraphics[height = 5cm]{LowerBound.jpg}
\caption{Deviation of the random variable $X$ below $Kt^m$.}
\label{Image:LowerBound}
\end{figure}
\fi
Now, we wish to obtain a tight lower bound that the random variable, $X$, deviates far below from the hazard rate of a manually tested software $Kt^m$. Here, for some $K>0, t>0$, and $m>-1$, the value $Kt^m$ is assumed to be below the expectation, $\mu_{\hat{z}}$, in a given time period $[0,t]$. That is, $Kt^m$ is the left marginal arbitrary constant from the expectation, $\mu_{\hat{z}}$. By equating the value of $(1-\delta)\mu$ with $z(t)$ in Equation \ref{Weibull-Hazard}, then we have:
\begin{align*}
    (1-\delta)\mu = Kt^m
\end{align*}
\begin{align}
\label{Weibull-Hazard-Delta}
    \Rightarrow \delta = 1 - \frac{Kt^m}{\mu}
\end{align}
Since, we are finding the tight lower bound for the deviation of a random variable ($X$) from the hazard rate of a manually tested software, in the above equation, we will replace $\mu$ with the expected hazard rate $\mu_{\hat{z}}$ of the software that is tested by using the SDP model. Now, we know from Equation \ref{Expectation}, the expected number of failures (which are assumed to be observed as a result of dormant defective modules) in a software which uses SDP model:
\begin{align*}
    \mu_{\hat{z}} = \mathbb{E}[X] = \sum_{i=1}^{l} \mathbb{E}[X_i] = \sum_{i=1}^{l} p = lp
\end{align*}
Now, substituting the value of $\delta$ (from Equation \ref{Weibull-Hazard-Delta}) and $\mu_{\hat{z}}$ (from Equation \ref{Expectation}) in Equation \ref{LowerChernoff} results in the tight lower bound for the deviation of a random variable far from the hazard rate of a manually tested software. This is shown below:
\begin{align}
    Pr[X < Kt^m] < e^{\frac{-\mu_{\hat{z}}\big[1-\frac{Kt^m}{\mu_{\hat{z}}}\big]^2}{2}}
\end{align}
A few steps of simplification will ensure the proof.
\end{proof}
The Theorem \ref{Theorem-Weibull-Hazard} provides evidence that the probability of occurrence of fewer hazards in SDP-based software is lower than the probability of occurrence of total hazards in the same software tested by a human and is bound by the maximum value of $e^{\frac{-(\mu_{\hat{z}}-Kt^m)^2}{2\mu_{\hat{z}}}}$. If we observe, the term in the exponent is negative. Hence, a large deviation between the terms such as $\mu_{\hat{z}}$, and $Kt^m$, in the numerator of the exponent will lead to more sharp bounds. In such a way, the term $e^{\frac{-(\mu_{\hat{z}}-Kt^m)^2}{2\mu_{\hat{z}}}}$ approaches 1 if the developed prediction model works equivalent to the quality of the testers. Similarly, the term $e^{\frac{-(\mu_{\hat{z}}-Kt^m)^2}{2\mu_{\hat{z}}}}$ approaches 0 if the developed prediction model work badly when compared with the quality of the testers.
\subsection{A Tight Upper Bound in terms of the Reliability of a Software}
\label{Bound on the Reliability}
In this section, we provide a lemma that calculates the reliability from the hazard rate of software that is tested using the SDP model. Here, as discussed above, the hazard rate of a software that is tested using SDP models is assumed to be a random variable, $X$.
\begin{lemma}
\label{Lemma-Reliability-X}
Let $X_1, X_2, \dots, X_l$ be the independent Bernoulli trials, then for the hazards in the software system, $\hat{z}(t)$ = X = $\sum_{i=1}^{l} X_i$, its reliability is:
\begin{align}
\label{Reliability-X}
    \hat{R}(t) = e^{-Xt}
\end{align}
\end{lemma}
\begin{proof}
We have that the hazards in a software (that is tested using SDP model) is a random variable $X$. Then using Equation \ref{Reliability from Hazard Rate}:
\begin{align}
\label{Reliability-X-Initial}
    \hat{R}(t) = e^{-\int_0^t X dx}
\end{align}
Here, we use $\hat{R}(t)$, to represent the reliability of the software that is tested using the SDP model. Now, simplifying Equation \ref{Reliability-X-Initial} will result in the accomplishment of the proof.
\end{proof}
Now, to find the tight upper bound for the deviation of a random variable, $e^{-Xt}$, far from the reliability of a manually tested software, $e^{\frac{-Kt^{m+1}}{m+1}}$, we require to compute the expected reliability of a SDP-based software. From Lemma \ref{Lemma-Reliability-X}, we know the reliability of a software that is tested by using the SDP model. Now the expected reliability ($\mu_{\hat{R}}$ or $\mathbb{E}[\hat{R}(t)]$) is derived as:
\begin{equation}
\mathbb{E}[\hat{R}(t)] = \mu_{\hat{R}} = \mathbb{E}[e^{-Xt}] = \mathbb{E}\Big[e^{-t\sum_{i=1}^{l} X_i}\Big]
\end{equation}
Since the random variables ($X_i$s) are assumed to be independent, the sum of the terms in the exponent will become the product of the exponential terms. This is given as:
\begin{equation}
\mathbb{E}\Big[e^{-t\sum_{i=1}^{l} X_i}\Big] = \prod_{i=1}^{l} \mathbb{E}\big[e^{-tX_i}\big]
\end{equation}
Here, the random variable, $e^{-tX_i}$ assumes a value of $e^{-t}$ with probability $p$ and the value 1 with probability $1-p$. Now, using these values, we have the following from the above equation:
\begin{equation}
\label{Expected Reliability-SDP-1}
\prod_{i=1}^{l} \mathbb{E}\big[e^{-tX_i}\big] = \prod_{i=1}^{l} \big[pe^{-t}+1-p\big] = \prod_{i=1}^{l} \big[p(e^{-t}-1)+1\big]
\end{equation}
We know that, $1+x < e^x$. Now, using this inequality with $x = p(e^{-t}-1)$, we rewrite Equation \ref{Expected Reliability-SDP-1} to obtain the expected reliability:
\begin{equation}
\label{Expected Reliability-SDP}
\mathbb{E}[\hat{R}(t)] = \mu_{\hat{R}} = \prod_{i=1}^{l} \big[p(e^{-t}-1)+1\big] < \prod_{i=1}^{l} e^{p(e^{-t}-1)} = e^{lp(e^{-t}-1)}
\end{equation}
Now, by using the lemmas \ref{Lemma-Reliability-Wibull} and \ref{Lemma-Reliability-X}, the following theorem defines the deviation of a random variable $e^{-Xt}$ above its expectation $\mu_{\hat{R}}$. 
\begin{theorem}
\label{Theorem-Weibull-Reliability}
Let $X_1, X_2, \dots, X_l$ be the independent Bernoulli trials such that for, $1\leq i \leq l$, $Pr[X_i = 1] = p$, where $0 < p < 1$. Then for X = $\sum_{i=1}^{l} X_i, \mu_{\hat{R}} < e^{lp(e^{-t}-1)}, \exists$ parameters $K>0, m>-1$, time $t>0$, and for the Weibull Reliability function, $e^{\frac{-Kt^{m+1}}{m+1}}: $
\begin{align}
    Pr\Big[e^{-Xt} > e^{\frac{-Kt^{m+1}}{m+1}}\Big] < e^{-\big[e^{lp(e^{-t}-1)}-\frac{Kt^m}{m+1}\big]^2\frac{1}{2e^{lp(e^{-t}-1)}}}
\end{align}
\end{theorem}
\begin{proof}

The proof for this theorem is very similar to the proof for the lower tail, as we saw in the Theorem \ref{Theorem-Weibull-Hazard}. As before,
\begin{multline}
\label{Reliability-Otherform}
     Pr\Big[e^{-Xt} > e^{\frac{-Kt^{m+1}}{m+1}}\Big] =  Pr\Big[e^{Xt} < e^{\frac{Kt^{m+1}}{m+1}}\Big] \\
     = Pr\Big[Xt < \frac{Kt^{m+1}}{m+1}\Big]
     = Pr\Big[X < \frac{Kt^m}{m+1}\Big]
\end{multline}
Now, we wish to obtain a tight lower bound for the random variable, $X$, deviates far from the value $\frac{Kt^m}{m+1}$. Here, for some $K>0, t>0$, and $m>-1$, the value $\frac{Kt^m}{m+1}$ is assumed to be below the expectation, $\mu_{\hat{R}}$, in a given time period $[0,t]$. Now, equating the $(1-\delta)\mu$ in \ref{LowerChernoff} with $\frac{Kt^m}{m+1}$ in Equation \ref{Reliability-Otherform}, then we have:
\begin{align*}
    (1-\delta)\mu = \frac{Kt^m}{m+1}
\end{align*}
\begin{align}
    \Rightarrow \delta = 1 - \frac{Kt^m}{(m+1)\mu}
\end{align}
Since, we are finding the tight upper bound for the deviation of a random variable ($e^{-Xt}$) from the reliability of a manually tested software, in the above equation, we will replace $\mu$ with the expected reliability $\mu_{\hat{R}}$ of the software that is tested by using the SDP models. Substitute the expected reliability of a software (from Equation \ref{Expected Reliability-SDP}) in the above equation to get the final value of $\delta$.
\begin{align}
\label{Weibull-Reliability-Delta}
    \Rightarrow \delta = 1 - \frac{Kt^me^{-lp(e^{-t}-1)}}{(m+1)}
\end{align}
\iffalse
\begin{figure}
\centering
\includegraphics[height=5cm]{UpperBound.jpg}
\caption{Deviation of the exponent form of a random variable ($e^{-Xt}$ or the reliability from the random variable $X$) above $e^{\frac{-Kt^{m+1}}{m+1}}$ is equivalent to the deviation of the random variable $X$ below $\frac{Kt^m}{m+1}$.}
\label{Image:UpperBound}
\end{figure}
\fi
Here, the value of $\delta$ represents the band value from the expectation, $\mu_{\hat{R}}$. Now substitute the value of $\delta$ (from Equation \ref{Weibull-Reliability-Delta}) in Equation \ref{LowerChernoff}, then:
\begin{multline}
    Pr\Big[e^{-Xt} > e^{\frac{-Kt^{m+1}}{m+1}}\Big] = \\ Pr\Big[X < \frac{Kt^m}{m+1}\Big] < e^{\frac{-e^{lp(e^{-t}-1)}\big[1 - \frac{Kt^me^{-lp(e^{-t}-1)}}{(m+1)}\big]^2}{2}}
\end{multline}
After simplification, we arrive at the end of the proof.
\end{proof}
The Theorem \ref{Theorem-Weibull-Reliability} shows that the probability of obtaining better reliability in a software that is tested with the SDP model than in same software that is tested by humans is bounded by the maximum value of  $e^{-\big[e^{lp(e^{-t}-1)}-\frac{Kt^m}{m+1}\big]^2\frac{1}{2e^{lp(e^{-t}-1)}}}$. Similar to the result of Theorem \ref{Theorem-Weibull-Hazard}, if we observe, the term in the exponent is negative. A large deviation between the terms such as $e^{-lp(e^{-t}-1)}$, and $\frac{Kt^m}{m+1}$ in the numerator of the exponent will lead to obtain even more sharp bounds. In such a way, the term $e^{-\big[e^{lp(e^{-t}-1)}-\frac{Kt^m}{m+1}\big]^2\frac{1}{2e^{lp(e^{-t}-1)}}}$ approaches 1 if the developed prediction model works equivalent to the quality of the testers. Similarly, the term $e^{-\big[e^{lp(e^{-t}-1)}-\frac{Kt^m}{m+1}\big]^2\frac{1}{2e^{lp(e^{-t}-1)}}}$ approaches 0 if the developed prediction model work poor when compared with the quality of the testers. Precisely, if the estimated reliability of the manually tested software deviates far from the expectation (of the random variable), then the bound becomes more tight.

The Theorems \ref{Theorem-Weibull-Hazard} and \ref{Theorem-Weibull-Reliability} provide a theoretical basis for investigating the feasibility of the developed SDP models in real-world testing environments.

By using Theorems \ref{Theorem-Weibull-Hazard} and \ref{Theorem-Weibull-Reliability}, in Section \ref{Supplements}, we provide more bounds (in terms of both the hazard rate and reliability) for the supplements of the Weibull distribution. All the proofs provided in Section \ref{Supplements} are used to understand the change in the behaviour of the bounds at various forms of the Weibull distribution.
\section{Supplementary Proofs}
\label{Supplements}
In this section, we provide proofs based on the supplements of the Weibull distribution. All the proofs provided in this section are based on the theorems \ref{Theorem-Weibull-Hazard} and \ref{Theorem-Weibull-Reliability}. In each sub-section, we provide a possible scenario to observe the specific functional form of the hazard rate. All the proofs provided in this section are used to analyse the feasibility of the developed prediction model (in terms of both the hazard rate and the reliability) when compared with some specific forms of the reliability indices, such as the hazard rate and the reliability. This aids in incorporating and analysing various hazard (also reliability) functions in order to generate sharp bounds. 
\subsection{The Non-Linearly Decreasing Hazard Model}
\label{The Non-Linearly Decreasing Hazard Model}
In majority cases, the hazards in software are typically represented as a non-linearly decreasing function of time \cite{lyu1996handbook}.  For example, from Figure \ref{RevisedBathTub}, it is observed that every upgrade of software leads to early failures in a software system. Hence, in the early stages, quality control and initial product testing teams will eliminate the substandard functionalities to avoid the higher hazard rate. After several iterations of testing, the hazards in the system will be reduced gradually to their minimal value. The simplest model that represents the non-linearly decreasing function is the inverse time function. That is, as time progresses, the hazards in a system are reduced gradually, until eventually the system has fewer hazards. It should be noted that any complex function can be substituted for the function we used to represent the non-linearly decreasing hazard model. 

In Section \ref{Sec: Non-Linearly Decreasing Hazard Model}, we provide a tighter lower bound for the possibility of achieving minimum hazards in a software that uses SDP than the hazards in a manually tested software. Similarly, in Section \ref{Sec: Non-Linearly Decreasing Reliability Model} we provide a tight upper bound for the possibility of achieving maximum reliability for a software that uses SDP, when compared with the reliability of a manually tested software. In both cases, the hazard rate in the system is assumed to be a non-linearly decreasing function of time.
\subsubsection{Bounds in terms of the Hazard Model}
\label{Sec: Non-Linearly Decreasing Hazard Model}
\begin{definition}
For any manually tested software system, and for any $K>0, t>0$, the hazards that follow a non-linearly decreasing function of time are defined as:
\begin{align}
\label{Eq: Non-Linearly Decreasing Hazard}
    z_{nld}(t) = \frac{K}{\sqrt{t}}, \text{ for some } K>0, \text{and time } t>0
\end{align}
\end{definition}
Now, the following corollary defines the deviation of a random variable $X$ below the value, $z_{nld}(t)$.
\begin{corollary}
\label{Corollary-Non-Linear-Decreasing-Hazard}
Let $X_1, X_2, \dots, X_l$ be the independent Bernoulli trials such that for, $1\leq i \leq l$, $Pr[X_i = 1] = p$, where $0 < p < 1$. Then for X = $\sum_{i=1}^{l} X_i, \mu_{\hat{z}} = \mathbb{E}[X] = \sum_{i=1}^{l} p = lp$, and $\exists$ parameters $K>0$, time $t>0$, and for the non-linearly decreasing hazard rate, $\frac{K}{\sqrt{t}}$:
\begin{align}
    Pr[X < \frac{K}{\sqrt{t}}] < e^{\frac{-(\sqrt{t}\mu_{\hat{z}}-K)^2}{2\mu_{\hat{z}} t}}
\end{align}
\end{corollary}
\begin{proof}
Substitute the value of $z(t)$ with the value $z_{nld}(t)$ (from Equation \ref{Eq: Non-Linearly Decreasing Hazard}) in Theorem \ref{Theorem-Weibull-Hazard}, then we get:
\begin{align}
    Pr[X < \frac{K}{\sqrt{t}}] < e^{\frac{-\mu_{\hat{z}}\Big[1-\frac{K}{\sqrt{t}\mu_{\hat{z}}}\Big]^2}{2}}
\end{align}
Simplifying the above equation will lead to the proof of the corollary.
\end{proof}
The Corollary \ref{Corollary-Non-Linear-Decreasing-Hazard} provides evidence that, for a non-linearly decreasing hazard model, the probability of the occurrence of fewer hazards in the software that uses SDP than the occurrence of the total number of hazards in the same software that is tested by humans is bounded by the maximum value of $e^{\frac{-(\sqrt{t}\mu_{\hat{z}}-K)^2}{2\mu_{\hat{z}} t}}$.
\subsubsection{Bounds in terms of the Reliability Model}
\label{Sec: Non-Linearly Decreasing Reliability Model}
In this section, we first define the reliability of a software that follows a non-linearly decreasing hazard model of the form, $\frac{K}{\sqrt{t}}$, before determining the bound. 
\begin{lemma}
\label{Lemma-Reliability-NLD}
For the non-linearly decreasing function of hazard model $z_{nld}(t) = \frac{K}{\sqrt{t}}, \text{ for some } K>0, \text{ and time } t>0$, its reliability is:
\begin{align}
\label{Reliability-NLD}
    R_{nld}(t) = e^{-2K\sqrt{t}}
\end{align}
\end{lemma}
\begin{proof}
We assume that the hazards in software are a non-linearly decreasing function of time. It is given as:
\begin{align*}
    z_{nld}(t) = \frac{K}{\sqrt{t}}, \text{for some } K>0, \text{and time } t>0
\end{align*}
Now, substituting $z_{nld}(t)$ in $R(t)$ (in Equation \ref{Reliability from Hazard Rate}) will results in:
\begin{align}
\label{Reliability-NLD-Inintial}
    R_{nld}(t) = e^{-\int_0^t \frac{K}{\sqrt{x}} dx}
\end{align}
Now, simplifying Equation \ref{Reliability-NLD-Inintial} will give us the reliability of a software that is estimated from the non-linearly decreasing function of the hazard model.
\end{proof}

\begin{corollary}
\label{Corollary-Non-Linear-Decreasing-Reliability}
Let $X_1, X_2, \dots, X_l$ be the independent Bernoulli trials such that for, $1\leq i \leq l$, $Pr[X_i = 1] = p$, where $0 < p < 1$. Then for X = $\sum_{i=1}^{l} X_i, \mu_{\hat{R}} < e^{lp(e^{-t}-1)}$, and $\exists$ parameters $K>0$, time $t>0$, and for the reliability of the software that follows a non-linearly decreasing hazard rate, $R_{nld}(t) = e^{-2K\sqrt{t}}$:
\begin{align}
    Pr\Big[e^{-Xt} > e^{-2K\sqrt{t}}\Big] < e^{\frac{-e^{lp(1-e^{-t})}}{2}\Big[e^{lp(e^{-t}-1)}-\frac{2K}{\sqrt{t}}\Big]^2}
\end{align}
\end{corollary}
\begin{proof}
Similar to Theorem \ref{Theorem-Weibull-Reliability}, the proof for this upper tail is very similar to the proof for the lower tail as we also saw in Corollary \ref{Corollary-Non-Linear-Decreasing-Hazard}. As before,
\begin{multline}
\label{NLD-Reliability-Formating}
     Pr\Big[e^{-Xt} > e^{-2K\sqrt{t}}\Big] = Pr\Big[-Xt > -2K\sqrt{t}\Big] \\ = Pr\Big[-X > -\frac{2K}{\sqrt{t}}\Big] = Pr\Big[X < \frac{2K}{\sqrt{t}}\Big]
\end{multline}

Here, we wish to obtain a tight lower bound that the random variable, $X$, deviates far from the value $\frac{2K}{\sqrt{t}}$. Here, for some $K>0$ and $t>0$, the value $\frac{2K}{\sqrt{t}}$ is assumed to be below the expectation, $\mu_{\hat{R}}$, in a given time period $[0,t]$. Now, equating the $(1-\delta)\mu$ in \ref{LowerChernoff} with $\frac{2K}{\sqrt{t}}$, then we have:
\begin{align*}
    (1-\delta)\mu = \frac{2K}{\sqrt{t}}
\end{align*}
\begin{align}
\label{NLD-Reliability-Delta}
    \Rightarrow \delta = 1 - \frac{2K}{\mu\sqrt{t}}
\end{align}
Since, we are finding the tight upper bound (we eventually converted it into finding the tight lower bound of some other form) for the deviation of a random variable ($e^{-Xt}$) from the reliability of a manually tested software, in the above equation, we will replace $\mu$ with the expected reliability $\mu_{\hat{R}}$ of the software that is tested by using the SDP models. Now, substitute the value of $\delta$ (from Equation \ref{NLD-Reliability-Delta}) in Equation \ref{LowerChernoff} to find the tight upper bound for the deviation of a random variable $e^{-Xt}$ far from the value $e^{-2K\sqrt{t}}$. This is given as: 
\begin{align}
    Pr\Big[X < \frac{2K}{\sqrt{t}}\Big] < e^{\frac{-\mu_{\hat{R}}\big[1-\frac{2K}{\sqrt{t}\mu_{\hat{R}}}\big]^2}{2}}
\end{align}
The below equation is the simplification of the above Equation.
\begin{align}
       Pr\Big[X < \frac{2K}{\sqrt{t}}\Big] = Pr\Big[e^{-Xt} > e^{-2K\sqrt{t}}\Big] < e^{-\frac{1}{2\mu_{\hat{R}}}\Big[\mu_{\hat{R}}-\frac{2K}{\sqrt{t}}\Big]^2}
\end{align}
Substituting the expected reliability, $\mu_{\hat{R}} = e^{lp(e^{-t}-1)}$ in the above equation will results in the end of the proof.
\end{proof}
The Corollary \ref{Corollary-Non-Linear-Decreasing-Reliability} shows that, in terms of reliability (that is calculated from a non-linearly decreasing hazard model), the probability of obtaining better reliability in the software that is tested using the SDP model than the same software that is tested by humans is bounded by the maximum value of $e^{\frac{-e^{lp(1-e^{-t})}}{2}\Big[e^{lp(e^{-t}-1)}-\frac{2K}{\sqrt{t}}\Big]^2}$.

\subsection{The Linearly Decreasing Hazard Model}
\label{The Linearly Decreasing Model}
Assume that the hazards in the software are detected and removed in a linear fashion. Even though this is a rare case, we are providing the proofs (in terms of both the hazard rate and reliability) for the linearly decreasing hazard model. The simplest linearly decreasing function of hazards is assumed to be a slope-intercept form of the equation of the line, where, at the negative gradient, the equation of the line becomes a linearly decreasing function of time $t$.

In Section \ref{Sec: Linearly Decreasing Hazard Model}, we provide a proof for the possibility of achieving minimum hazards in a software that uses SDP compared to the hazards in a manually tested software. Similarly, in Section \ref{Sec: Linearly Decreasing Reliability Model} we provide a proof for the possibility of achieving maximum reliability for a software that uses SDP, when compared with the reliability of a manually tested software. In both cases, the hazard rate in the system is assumed to be a linearly decreasing function of time.
\subsubsection{Bounds in terms of the Hazard Model}
\label{Sec: Linearly Decreasing Hazard Model}
\begin{definition}
For any manually tested software system and for any $K>0$ and $m>0$, the hazards that follow a linearly decreasing function of time are defined as:
\begin{align}
\label{Eq: Linearly Decreasing Hazard}
    z_{ld}(t) = K-mt, \text{ for some } K>0, m>0 \text{ and time } t>0
\end{align}
\end{definition}

Now, the following corollary defines the deviation of a random variable $X$ below the value, $z_{ld}(t)$.
\begin{corollary}
\label{Corollary-Linearly-Decreasing-Hazard}
Let $X_1, X_2, \dots, X_l$ be the independent Bernoulli trials such that for, $1\leq i \leq l$, $Pr[X_i = 1] = p$, where $0 < p < 1$. Then for X = $\sum_{i=1}^{l} X_i, \mu_{\hat{z}} = \mathbb{E}[X] = \sum_{i=1}^{l} p = lp$, and $\exists$ parameters $K>0, m>0$, time $t>0$, and for the linearly decreasing hazard rate, $K-mt$:
\begin{align}
    Pr[X < K-mt] < e^{\frac{-(\mu_{\hat{z}}-K+mt)^2}{2\mu_{\hat{z}}}}
\end{align}
\end{corollary}
\begin{proof}
Substitute the value of $z(t)$ with the value $z_{ld}(t)$ (from Equation \ref{Eq: Linearly Decreasing Hazard}) in Theorem \ref{Theorem-Weibull-Hazard}, then we get:
\begin{align}
    Pr[X < K-mt] < e^{\frac{-\mu_{\hat{z}}\Big[1-\frac{K-mt}{\mu_{\hat{z}}}\Big]^2}{2}}
\end{align}
A few steps of simplification results in deriving the tight bound for the deviation of the random variable $X$ below the value, $z_{ld}(t)$.
\end{proof}
The Corollary \ref{Corollary-Linearly-Decreasing-Hazard} provides an evidence that, for a linearly decreasing hazard model, the probability of occurrence of less hazards in the software that uses SDP, than the occurrence of the total hazards in the same software, that is tested by the human, is bounded by the maximum value of $e^{\frac{-(\mu_{\hat{z}}-K+mt)^2}{2\mu_{\hat{z}}}}$.
\subsubsection{Bounds in terms of the Reliability Model}
\label{Sec: Linearly Decreasing Reliability Model}
In the following lemma, we define the reliability of a software, that is calculated from a linearly decreasing function of hazard rate.
\begin{lemma}
\label{Lemma-Reliability-LD}
For the linearly decreasing function of hazard model $z_{ld}(t) = K-mt, \text{ for some } K>0, m>0 \text{ and time } t>0$, its reliability is:
\begin{align}
\label{Reliability-LD}
    R_{ld}(t) = e^{\big[\frac{mt^2}{2}-Kt\big]}
\end{align}
\end{lemma}
\begin{proof}
We assume that the hazards in a manually-tested software are a linearly decreasing function of time. It is given as:
\begin{align*}
    z_{ld}(t) = K-mt, \text{ for some } K>0, m>0 \text{ and time } t>0
\end{align*}
Now, substituting $z_{ld}(t)$ in $R(t)$ (from Equation \ref{Reliability from Hazard Rate}) will results in:
\begin{align}
\label{Reliability-LD-Inintial}
    R_{ld}(t) = e^{-\int_0^t K-mx \text{ } dx}
\end{align}
Now, simplifying Equation \ref{Reliability-LD-Inintial} will result in the proof.
\end{proof}
\begin{corollary}
\label{Corollary-Linearly-Decreasing-Reliability}
Let $X_1, X_2, \dots, X_l$ be the independent Bernoulli trials such that for, $1\leq i \leq l$, $Pr[X_i = 1] = p$, where $0 < p < 1$. Then for X = $\sum_{i=1}^{l} X_i, \mu_{\hat{R}} < e^{lp(e^{-t}-1)}$, and $\exists$ parameters $K>0, m>0$, time $t>0$, and for the reliability of the software that follows a linearly decreasing hazard rate, $R_{ld}(t) = e^{\big[\frac{mt^2}{2}-Kt\big]}$:
\begin{align}
        Pr\Big[e^{-Xt} > e^{\big[\frac{mt^2}{2}-Kt\big]}\Big] < e^{\frac{-e^{lp(1-e^{-t})}[2e^{lp(e^{-t}-1)}-2K+mt]^2}{8}}
\end{align}
\end{corollary}
\begin{proof}
The proof for this upper tail is very similar to the proof for the lower tail, as we saw in Corollary \ref{Corollary-Linearly-Decreasing-Hazard}. As before,
\begin{multline}
\label{LD-Reliability-Formating}
     Pr\Big[e^{-Xt} > e^{\big[\frac{mt^2}{2}-Kt\big]}\Big] = Pr\Big[Xt < Kt-\frac{mt^2}{2}\Big] = \\ Pr\Big[X < K-\frac{mt}{2}\Big]
\end{multline}
Here, we wish to obtain a tight lower bound that the random variable, $X$, deviates far from the value $K-\frac{mt}{2}$. Here, for some $K>0, m>0$ and $t>0$, the value $K-\frac{mt}{2}$ is assumed to be below the expectation, $\mu_{\hat{R}}$, in a given time period $[0,t]$. Now, equating the $(1-\delta)\mu$ in \ref{LowerChernoff} with $K-\frac{mt}{2}$ (in \ref{LD-Reliability-Formating}), then we have:
\begin{align*}
    (1-\delta)\mu = K-\frac{mt}{2}
\end{align*}
\begin{align}
\label{LD-Reliability-Delta}
    \Rightarrow \delta = 1 - \frac{2K-mt}{2\mu}
\end{align}
Since, we are finding the tight upper bound (we eventually converted it into finding the tight lower bound of some other form) for the deviation of a random variable ($e^{-Xt}$) from the reliability of a manually tested software, in the above equation, we will replace $\mu$ with the expected reliability $\mu_{\hat{R}}$ of the software that is tested by using the SDP models. Now, substitute the value of $\delta$ (from Equation \ref{LD-Reliability-Delta}) in Equation \ref{LowerChernoff} to find the tight upper bound for the deviation of a random variable $X$ far from the value $K-\frac{mt}{2}$. This is given as: 
\begin{align}
    Pr\Big[X < K-\frac{mt}{2}\Big] < e^{\frac{-\mu_{\hat{R}}\big[1-\frac{2K-mt}{2\mu_{\hat{R}}}\big]^2}{2}}
\end{align}
Few steps of simplification will provide the following proof: 
\begin{align}
        Pr\Big[e^{-Xt} > e^{\big[\frac{mt^2}{2}-Kt\big]}\Big] < e^{-\frac{[2\mu_{\hat{R}}-2K+mt]^2}{8\mu_{\hat{R}}}}
\end{align}
Now substituting the expected reliability $\mu_{\hat{R}} = e^{lp(e^{-t}-1)}$ in the above equation will satisfies the Corollary.

\end{proof}
The Corollary \ref{Corollary-Linearly-Decreasing-Reliability} provides evidence that, in terms of reliability (that is calculated from a non-linearly decreasing hazard model), the probability of obtaining better reliability in the software that is tested by using the SDP model than the same software that is tested by humans is bounded by the maximum value of $e^{\frac{-e^{lp(1-e^{-t})}[2e^{lp(e^{-t}-1)}-2K+mt]^2}{8}}$.

\subsection{The Non-Linearly Increasing Hazard Model}
\label{The Non-Linearly Increasing Hazard Model}
Hazards may increase in the system due to the cascading effect \cite{lyu1996handbook}. In such cases, the failure of one module leads to the failure of the subsequent modules. This results in a non-linear (in fact, even exponential) increase in the hazards in the software system. The simplest non-linearly increasing hazard model that can be postulated is one in which the hazards increase in a quadratic fashion with time. The following sub-sections provide the proofs that are based on the hazard model that follows the non-linearly increasing hazard model.

In Section \ref{Sec: Non-Linearly Increasing Hazard Model}, we provide a proof for the possibility of achieving minimum hazards in a software that uses SDP compared to the hazards in a manually tested software. Similarly, in Section \ref{Sec: Non-Linearly Increasing Reliability Model} we provide a proof for the possibility of achieving maximum reliability for a software that uses SDP, when compared with the reliability of a manually tested software.
\subsubsection{Bounds in terms of the Hazard Model}
\label{Sec: Non-Linearly Increasing Hazard Model}
\begin{definition}
For any manually tested software system, and for any $K>0$, the hazards that follow a non-linearly increasing function of time are defined as:
\begin{align}
\label{Eq: Non-Linearly Increasing Hazard}
    z_{nli}(t) = Kt^2, \text{ for some } K>0 \text{ and time } t>0
\end{align}
\end{definition}
Similar to the other supplements, Equation \ref{Eq: Non-Linearly Increasing Hazard} is the special case of the Weibull hazard rate where, at the value of $m=2$, the hazard rate becomes a quadratic function of time. Now, the following corollary defines the deviation of a random variable, $X$ below the value, $Kt^2$.
\begin{corollary}
\label{Corollary-Non-Linearly-Increasing-Hazard}
Let $X_1, X_2, \dots, X_l$ be the independent Bernoulli trials such that for, $1\leq i \leq l$, $Pr[X_i = 1] = p$, where $0 < p < 1$. Then for X = $\sum_{i=1}^{l} X_i, \mu_{\hat{z}} = \mathbb{E}[X] = \sum_{i=1}^{l} p = lp$, and $\exists$ parameters $K>0$, time $t>0$, and for the non-linearly increasing hazard rate $z_{nli} = Kt^2$:
\begin{align}
    Pr[X < Kt^2] < e^{\frac{-(\mu_{\hat{z}}-Kt^2)^2}{2\mu_{\hat{z}}}}
\end{align}
\end{corollary}
\begin{proof}
Substitute the value of $z(t)$ with the value $z_{nli}(t)$ (from Equation \ref{Eq: Non-Linearly Increasing Hazard}) in Theorem \ref{Theorem-Weibull-Hazard}, then we get:
\begin{align}
    Pr[X < Kt^2] < e^{\frac{-\mu_{\hat{z}}\Big[1-\frac{Kt^2}{\mu_{\hat{z}}}\Big]^2}{2}}
\end{align}
A simplification of the above equation will ensure the proof.
\end{proof}
The Corollary \ref{Corollary-Non-Linearly-Increasing-Hazard} provides evidence that, in terms of the non-linearly increasing hazard model, the probability of occurrence of fewer hazards in the software that uses SDP than the occurrence of the total hazards in the same software that is tested by humans is bound by the maximum value of $e^{\frac{-(\mu_{\hat{z}}-Kt^2)^2}{2\mu_{\hat{z}}}}$.
\subsubsection{Bounds in terms of the Reliability Model}
\label{Sec: Non-Linearly Increasing Reliability Model}
In the following lemma, we define the reliability of a software that follows a non-linearly increasing model of the hazard rate.
\begin{lemma}
\label{Lemma-Reliability-NLI}
For the non-linearly increasing function of the hazard model, $z_{nli}(t) = Kt^2, \text{ for some } K>0, \text{ and time } t>0$, its reliability is:
\begin{align}
\label{Reliability-NLI}
    R_{nli}(t) = e^{-\frac{Kt^3}{3}}
\end{align}
\end{lemma}
\begin{proof}
We assume that the hazards in software are a linearly increasing function of time. It is given as:
\begin{align*}
    z_{nli}(t) = Kt^2, \text{ for some } K>0 \text{ and time } t>0
\end{align*}
Now, substituting $z_{nli}(t)$ in $R(t)$ (from Equation \ref{Reliability from Hazard Rate}) will results in:
\begin{align}
\label{Reliability-NLI-Inintial}
    R_{nli}(t) = e^{-\int_0^t Kx^2 dx}
\end{align}
Now, simplifying Equation \ref{Reliability-NLI-Inintial} will results in the reliability of a software that has a non-linearly increasing function of the hazard rate.
\end{proof}
\begin{corollary}
\label{Corollary-Non-Linearly-Increasing-Reliability}
Let $X_1, X_2, \dots, X_l$ be the independent Bernoulli trials such that for, $1\leq i \leq l$, $Pr[X_i = 1] = p$, where $0 < p < 1$. Then for X = $\sum_{i=1}^{l} X_i, \mu_{\hat{R}} < e^{lp(e^{-t}-1)}$, and $\exists$ parameters $K>0$, time $t>0$, and for the reliability of the software that follows a non-linearly increasing hazard rate, $R_{nli}(t) = e^{\big[-\frac{Kt^3}{3}\big]}$:
\begin{align}
         Pr\Big[e^{-Xt} > e^{\big[-\frac{Kt^3}{3}\big]}\Big] < e^{\frac{-e^{lp(1-e^{-t})}\big[3e^{lp(e^{-t}-1)}-Kt^2\big]^2}{18}}
\end{align}
\end{corollary}
\begin{proof}
The proof for this upper tail is very similar to the proof for the lower tail, as we saw in Corollary \ref{Corollary-Non-Linearly-Increasing-Hazard}. As before,
\begin{align}
\label{NLI-Reliability-Formating}
     Pr\Big[e^{-Xt} > e^{\big[-\frac{Kt^3}{3}\big]}\Big] = Pr\Big[Xt < \frac{Kt^3}{3}\Big] = Pr\Big[X < \frac{Kt^2}{3}\Big]
\end{align}
Now, we wish to obtain a tight lower bound that the random variable, $X$, deviates far from the value $\frac{Kt^2}{3}$. Here, for some $K>0$ and $t>0$, the value $\frac{Kt^2}{3}$ is assumed to be below the expectation, $\mu_{\hat{R}}$, in a given time period $[0,t]$. Now, equating the $(1-\delta)\mu$ in \ref{LowerChernoff} with $\frac{Kt^2}{3}$ in \ref{NLI-Reliability-Formating}, then we have:
\begin{align*}
    (1-\delta)\mu = \frac{Kt^2}{3}
\end{align*}
\begin{align}
\label{NLI-Reliability-Delta}
    \Rightarrow \delta = 1 - \frac{Kt^2}{3\mu}
\end{align}
Since, we are finding the tight upper bound (we eventually converted it into finding the tight lower bound of some other form) for the deviation of a random variable ($e^{-Xt}$) from the reliability of a manually tested software, in the above equation, we will replace $\mu$ with the expected reliability $\mu_{\hat{R}}$ of the software that is tested by using the SDP models. Now, substitute the value of $\delta$ (from Equation \ref{NLI-Reliability-Delta}) in Equation \ref{LowerChernoff} to find the tight upper bound for the deviation of a random variable $X$ far from the value $\frac{Kt^2}{3}$. This is given as: 
\begin{align}
    Pr\Big[X < \frac{Kt^2}{3}\Big] < e^{\frac{-\mu_{\hat{R}}\big[1-\frac{Kt^2}{3\mu_{\hat{R}}}\big]^2}{2}}
\end{align}
Few steps of simplification will lead to: 
\begin{align}
         Pr\Big[e^{-Xt} > e^{\big[-\frac{Kt^3}{3}\big]}\Big] < e^{-\frac{[3\mu_{\hat{R}}-Kt^2]^2}{18\mu_{\hat{R}}}}
\end{align}
Substituting the expected reliability $\mu_{\hat{R}} = e^{lp(e^{-t}-1)}$ in the above equation will ensure the proof.
\end{proof}
The Corollary \ref{Corollary-Non-Linearly-Increasing-Reliability} provides evidence that, in terms of reliability of a software that is calculated from a non-linearly increasing hazard model, the probability of obtaining better reliability in the software that is tested by using the SDP model than the reliability of the same software that is tested by humans is bounded by the maximum value of $e^{\frac{-e^{lp(1-e^{-t})}\big[3e^{lp(e^{-t}-1)}-Kt^2\big]^2}{18}}$.

\subsection{The Linearly Increasing Hazard Model}
\label{The Linearly Increasing Model}
In general, when wear, deterioration, or feature upgradation occur in the system, the hazards will increase as the time passes \cite{pressman2005software, hartz1997introduction}. For example, when feature upgrades occur in the system, since the software experiences an enhancement in functionality, the complexity of the software is likely to increase. Consequently, it is likely to incur hazards in the system due to the unintended injection of the defective modules \cite{lyu1996handbook}. The simplest increasing hazard model that can be postulated is one in which the hazards in a system increase linearly with time. The following sub-sections provide the proofs that are based on the hazard model that follows the linearly increasing hazard model.

In Section \ref{Sec: Linearly Increasing Hazard Model}, we provide a proof for the possibility of achieving minimum hazards in a software that uses SDP compared to the hazards in a manually tested software. Similarly, in Section \ref{Sec: Linearly Increasing Reliability Model} we provide a proof for the possibility of achieving maximum reliability for a software that uses SDP, when compared with the reliability of a manually tested software.
\subsubsection{Bounds in terms of the Hazard Model}
\label{Sec: Linearly Increasing Hazard Model}
\begin{definition}
For any manually tested software system, and for any $K>0$, the hazards that follow a linearly increasing function of time are defined as:
\begin{align}
\label{Eq: Linearly Increasing Hazard}
    z_{li}(t) = Kt, \text{ for some } K>0 \text{ and time } t>0
\end{align}
\end{definition}

Now, the following corollary defines the deviation of a random variable $X$ below the value $Kt$.
\begin{corollary}
\label{Corollary-Linearly-Increasing-Hazard}
Let $X_1, X_2, \dots, X_l$ be the independent Bernoulli trials such that for, $1\leq i \leq l$, $Pr[X_i = 1] = p$, where $0 < p < 1$. Then for X = $\sum_{i=1}^{l} X_i, \mu_{\hat{z}} = \mathbb{E}[X] = \sum_{i=1}^{l} p = lp$, and $\exists$ parameters $K>0$, time $t>0$, and for the linearly increasing hazard rate $z_{li} = Kt$:
\begin{align}
    Pr[X < Kt] < e^{\frac{-(\mu_{\hat{z}}-Kt)^2}{2\mu_{\hat{z}}}}
\end{align}
\end{corollary}
\begin{proof}
Substitute the value of $z(t)$ with the value $z_{li}(t)$ (from Equation \ref{Eq: Linearly Increasing Hazard}) in Theorem \ref{Theorem-Weibull-Hazard}, then we get:
\begin{align}
    Pr[X < Kt] < e^{\frac{-\mu_{\hat{z}}\Big[1-\frac{Kt}{\mu_{\hat{z}}}\Big]^2}{2}}
\end{align}
Simplifying the above equation will result in the final form of the bound.
\end{proof}
The Corollary \ref{Corollary-Linearly-Increasing-Hazard} shows that, for a linearly increasing hazard model, the probability of occurrence of less hazards in the software that uses SDP, than the occurrence of the total hazards in the same software, that is tested by the human, is bounded by the maximum value of $e^{\frac{-(\mu_{\hat{z}}-Kt)^2}{2\mu_{\hat{z}}}}$. If the numerator in the exponent approaches infinite, then we get more sharp bounds.
\subsubsection{Bounds in terms of the Reliability Model}
\label{Sec: Linearly Increasing Reliability Model}
\begin{lemma}
\label{Lemma-Reliability-LI}
For the linearly increasing function of the hazard model $z_{li}(t) = Kt, \text{ for some } K>0, \text{ and time } t>0$, its reliability is:
\begin{align}
\label{Reliability-LI}
    R_{li}(t) = e^{-\frac{Kt^2}{2}}
\end{align}
\end{lemma}
\begin{proof}
We assume that the hazards in software are a linearly increasing function of time. It is given as:
\begin{align*}
    z_{li}(t) = Kt, \text{ for some } K>0 \text{ and time } t>0
\end{align*}
Now, substituting $z_{li}(t)$ in $R(t)$ (from Equation \ref{Reliability from Hazard Rate}) will results in:
\begin{align}
\label{Reliability-LI-Inintial}
    R_{li}(t) = e^{-\int_0^t Kx dx}
\end{align}
Now, simplifying Equation \ref{Reliability-LI-Inintial} will result in the reliability of a software.
\end{proof}

\begin{corollary}
\label{Corollary-Linearly-Increasing-Reliability}
Let $X_1, X_2, \dots, X_l$ be the independent Bernoulli trials such that for, $1\leq i \leq l$, $Pr[X_i = 1] = p$, where $0 < p < 1$. Then for X = $\sum_{i=1}^{l} X_i, \mu_{\hat{R}} < e^{lp(e^{-t}-1)}$, and $\exists$ parameters $K>0$, time $t>0$, and for the reliability of the software that follows a linearly increasing hazard rate, $R_{li}(t) = e^{-\frac{Kt^2}{2}}$:
\begin{align}
         Pr\Big[e^{-Xt} > e^{-\frac{Kt^2}{2}}\Big] =  Pr\Big[X < \frac{Kt}{2}\Big] < e^{-\frac{[2e^{lp(e^{-t}-1)}-Kt]^2}{8e^{lp(e^{-t}-1)}}}
\end{align}
\end{corollary}
\begin{proof}
The proof for this upper tail is very similar to the proof for the lower tail, as we saw in Corollary \ref{Corollary-Linearly-Increasing-Hazard}. As before,
\begin{align}
\label{LI-Reliability-Formating}
     Pr\Big[e^{-Xt} > e^{-\frac{Kt^2}{2}}\Big] = Pr\Big[Xt < \frac{Kt^2}{2}\Big] = Pr\Big[X < \frac{Kt}{2}\Big]
\end{align}
Now, we wish to obtain a tight lower bound that the random variable, $X$, deviates far from the value $\frac{Kt}{2}$. Here, for some $K>0$ and $t>0$, the value $\frac{Kt}{2}$ is assumed to be below the expectation, $\mu_{\hat{R}}$, in a given time period $[0,t]$. Now, equating the $(1-\delta)\mu$ in \ref{LowerChernoff} with $\frac{Kt}{2}$ in \ref{LI-Reliability-Formating}, then we have:
\begin{align*}
    (1-\delta)\mu = \frac{Kt}{2}
\end{align*}
\begin{align}
\label{LI-Reliability-Delta}
    \Rightarrow \delta = 1 - \frac{Kt}{2\mu}
\end{align}
Since, we are finding the tight upper bound (we eventually converted it into finding the tight lower bound of some other form) for the deviation of a random variable ($e^{-Xt}$) from the reliability of a manually tested software, in the above equation, we will replace $\mu$ with the expected reliability $\mu_{\hat{R}}$ of the software that is tested by using the SDP models. Now, substitute the value of $\delta$ (from Equation \ref{LI-Reliability-Delta}) in Equation \ref{LowerChernoff} to find the tight upper bound for the deviation of a random variable $X$ far from the value $\frac{Kt}{2}$. This is given as: 
\begin{align}
    Pr\Big[X < \frac{Kt}{2}\Big] < e^{\frac{-\mu_{\hat{R}}\big[1-\frac{Kt}{2\mu_{\hat{R}}}\big]^2}{2}}
\end{align}
Simplification of the above equation will lead to: 
\begin{align}
         Pr\Big[e^{-Xt} > e^{-\frac{Kt^2}{2}}\Big] < e^{-\frac{[2\mu_{\hat{R}}-Kt]^2}{8\mu_{\hat{R}}}}
\end{align}
Substituting the expected reliability $\mu_{\hat{R}} = e^{lp(e^{-t}-1)}$ into the above equation will result in the end of the proof.
\end{proof}
The Corollary \ref{Corollary-Linearly-Increasing-Reliability} provides evidence that, in terms of the reliability of a software that is calculated from a linearly increasing hazard model, the probability of obtaining better reliability in the software that is tested by using the SDP model than the reliability of the same software that is tested by humans is bounded by the maximum value of $e^{-\frac{[2e^{lp(e^{-t}-1)}-Kt]^2}{8e^{lp(e^{-t}-1)}}}$.

\subsection{The Constant Hazard Model}
\label{The Constant Hazard Model}
The region III of the revised bath-tub curve (from Figure \ref{RevisedBathTub}) for software reliability indicates the constant hazard rate in the system. That is, the software does not have any feature upgrades like the hardware does. In such a case, the system may approach obsolescence, indicating there are not enough incentives to upgrade it. As a result, the system does not produce any new failures. Hence, the hazard rate in the system remains constant over an infinite time period \cite{hartz1997introduction}. Even though the software does not accept any new upgrades, if it is in use, it does not produce any new failures (assuming the software always contains some fixed number of harmless failure incidents). In such a case, the hazard rate in the system will be some arbitrary constant, which is represented as $\lambda$ \cite{lyu1996handbook} \cite{hartz1997introduction}.

In Section \ref{Sec: Constant Hazard Model}, we provide a proof for the possibility of achieving minimum hazards in a software that uses SDP compared to the hazards in a manually tested software. Similarly, in Section \ref{Sec: Constant Reliability Model}, we provide a proof for the possibility of achieving maximum reliability for a software that uses SDP when compared with the reliability of a manually tested software. In both cases, the hazard rate in the system is assumed to be a constant value, and hence, it is independent of time.
\subsubsection{Bounds on the Hazard Model}
\label{Sec: Constant Hazard Model}
\begin{definition}
The constant hazards model is defined for any manually tested software system as follows:
\begin{align}
\label{Eq: Constant Hazard}
    z_{c}(t) = \lambda
\end{align}
\end{definition}

Now, the following corollary defines the deviation of a random variable $X$ below its expectation $\mu_{\hat{z}}$.
\begin{corollary}
\label{Corollary-Constant-Hazard}
Let $X_1, X_2, \dots, X_l$ be the independent Bernoulli trials such that for, $1\leq i \leq l$, $Pr[X_i = 1] = p$, where $0 < p < 1$. Then for X = $\sum_{i=1}^{l} X_i, \mu_{\hat{z}} = \mathbb{E}[X] = \sum_{i=1}^{l} p = lp$, and for the constant hazard rate $z_c = \lambda$:
\begin{align}
    Pr[X < \lambda] < e^{\frac{-(\mu_{\hat{z}}-\lambda)^2}{2\mu_{\hat{z}}}}
\end{align}
\end{corollary}
\begin{proof}
Substituting the value of $z(t)$ with the value of $z_{c}(t)$ in Theorem \ref{Theorem-Weibull-Hazard}, then we get:
\begin{align}
    Pr[X < \lambda] < e^{\frac{-\mu_{\hat{z}}\big[1-\frac{\lambda}{\mu_{\hat{z}}}\big]^2}{2}}
\end{align}
A few steps of simplification will lead to the end of the proof.
\end{proof}
The Corollary \ref{Corollary-Constant-Hazard} shows that, in terms of a constant hazard model, the probability of occurrence of fewer hazards in the software that uses SDP than the occurrence of the total hazards in the same software that is tested by humans is bounded by the maximum value of
$e^{\frac{-(\mu_{\hat{z}}-\lambda)^2}{2\mu_{\hat{z}}}}$.
\subsubsection{Bounds on the Reliability Model}
\label{Sec: Constant Reliability Model}
\begin{lemma}
\label{Lemma-Reliability-C}
For the constant hazard model $z_{c}(t) = \lambda$, its reliability is:
\begin{align}
\label{Reliability-C}
    R_{c}(t) = e^{-\lambda t}
\end{align}
\end{lemma}
\begin{proof}
We assume that at any time, $t>0$, the number of hazards in a software is a constant value. It is given as:
\begin{align*}
    z_{c}(t) = \lambda
\end{align*}
Now, substituting $z_{c}(t)$ in $R(t)$ (from Equation \ref{Reliability from Hazard Rate}) will results in:
\begin{align}
\label{Reliability-C-Inintial}
    R_{c}(t) = e^{-\int_0^t \lambda dx}
\end{align}
where $R_c(t)$ is the reliability of the software that follows a constant hazard model. Now, simplifying Equation \ref{Reliability-C-Inintial} will lead to the reliability of a software that has a constant hazard model.
\end{proof}

\begin{corollary}
\label{Corollary-Constant-Reliability}
Let $X_1, X_2, \dots, X_l$ be the independent Bernoulli trials such that for, $1\leq i \leq l$, $Pr[X_i = 1] = p$, where $0 < p < 1$. Then for X = $\sum_{i=1}^{l} X_i, \mu_{\hat{R}} < e^{lp(e^{-t}-1)}$, and for the reliability of the software that follows a constant hazard rate, $R_{c}(t) = e^{-\lambda t}$:
\begin{align}
    Pr\Big[e^{-Xt} > e^{-\lambda t}\Big] < e^{\frac{-(e^{lp(e^{-t}-1)}-\lambda)^2}{2e^{lp(e^{-t}-1)}}}
\end{align}
\end{corollary}
\begin{proof}
The proof for this upper tail is very similar to the proof for the lower tail as we saw in Corollary \ref{Corollary-Constant-Hazard}. From the Lemmas \ref{Lemma-Reliability-X} and \ref{Lemma-Reliability-C}, as before,
\begin{align}
\label{C-Reliability-Formating}
     Pr\Big[e^{-Xt} > e^{-\lambda t}\Big] = Pr\Big[Xt < \lambda t\Big] = Pr\Big[X < \lambda\Big]
\end{align}
The proof follows from Corollary \ref{Corollary-Constant-Hazard} but, at the expected reliability of $\mu_{\hat{R}} = e^{lp(e^{-t}-1)}$. From the proof of Corollary \ref{Corollary-Constant-Hazard}, we have that:
\begin{align*}
    Pr[X < \lambda] < e^{\frac{-(\mu_{\hat{z}}-\lambda)^2}{2\mu_{\hat{z}}}}
\end{align*}
Substituting the expected reliability $\mu_{\hat{R}}$ in place of $\mu_{\hat{z}}$ in the above equation will result in the tight lower bound.
\end{proof}
The Corollary \ref{Corollary-Constant-Reliability} provides evidence that, in terms of the reliability of a software that is calculated from a constant hazard model, the probability of obtaining better reliability in the software that is tested by using the SDP model than the reliability of the same software that is tested by humans is bounded by the maximum value of $e^{\frac{-(e^{lp(e^{-t}-1)}-\lambda)^2}{2e^{lp(e^{-t}-1)}}}$.
\section{Discussion}
\label{Discussion}
The proofs provided in Sections \ref{The Proofs} and \ref{Supplements} are based on the assumption that the random variable takes a value in the set, \{0,1\}, indicating that any defective component (a module), if that is predicted as clean, may cause a single failure in the system. Extending this assumption, in Section \ref{On the Definition of the Random Variable}, we discuss the chances of plausible amendments to the definition of the random variable. Precisely, we define some more bounds if each dormant defective module generates a specific form of Weibull distribution of the hazards instead of it generates a single failure in the software system.

Also, in Sections \ref{The Proofs} and \ref{Supplements}, we have provided various proofs for the deviation of the random variable (represented in terms of both hazard rate and the reliability of a software that is tested by using the SDP models) that deviates far from the actual values (represented in terms of both hazard rate and the reliability of a manually tested software). Now, in Section \ref{Analysis of the Theorems 1 to 4}, we analyse the theorems in light of the parameters used in the constructed proofs. 

%The proofs constructed in Sections \ref{The Proofs} and \ref{Supplements} are based on the assumptions that we stated in Section \ref{Preliminaries}. In Section \ref{On the Posed Assumptions}, we provide insights in relaxing the stated assumptions.
% and Section \ref{Analysis of the Corollaries 1-10} discusses the range of values that the bounds of the Corollaries \ref{Corollary-Non-Linear-Decreasing-Hazard} to \ref{Corollary-Constant-Reliability} take.
\subsection{On the Definition of the Random Variable}
\label{On the Definition of the Random Variable}
All the proofs provided in Sections \ref{The Proofs} and \ref{Supplements} are based on the assumption that the random variable takes the value of either 1 or 0. In other words, if the random variable takes the value 1, it indicates that a false negative module (that contains a dormant defect) can cause a single failure in the system. Extending this, in this section, we provide the basis to derive numerous proofs based on the assumption that the random variable takes a value based on the function of the parameters such as $\hat{K}, \hat{m}$, and $t$, instead of simply taking either 1 or 0. The following assumption ensures the same meaning:
\begin{assumption}
\label{Assumption10}
In any software, each misclassified defect module follows a Weibull distribution of the hazard form, $\hat{K}t^{\hat{m}}$.
\end{assumption}
All together, from Assumption \ref{Assumption10}, we calculate the hazard rate of the software system that is summed up from the hazard rates of the individual components (software modules). Formally, for the same newly developed software $S$ with $n$ modules, let us define a random variable $Y$ that takes a value $\hat{K}t^{\hat{m}}$ if the module $M_i$ is classified wrongly into the clean class, and takes 0 if the module $M_i$ is classified correctly into the clean class. This is given as:
\begin{equation}
\label{IndicatorRV-Yi}
Y_i = \begin{dcases*}
\hat{K}t^{\hat{m}}, & if $M_i$ is classified wrongly into clean class. \\
0,                  & if $M_i$ is classified correctly into clean class.
\end{dcases*}
\end{equation}
Here, we are assuming a similar type of distribution for each misclassified defective software module. In Equation \ref{IndicatorRV-Yi}, for the parameters $\hat{K}>0, \hat{m}>-1$, and $t>0$, the value $\hat{K}t^{\hat{m}}$ indicates the Weibull distribution for the random variable $Y_i$.

From the Assumption \ref{Assumption2}, since we set the fact that an organisation utilises the services of the traditional SDP model (that is developed using batch learning), and from Assumption \ref{Assumption3}, we ensure a similar probability for the predicted value of each newly developed module, then similar to the case of defining the probability value for the random variable $X_i$, in this case, each predicted clean module going into the wrong class is assigned with the probability, $p$. This is represented as:
\begin{equation}
\label{Probability-Y}
    \textbf{Pr}[Y_i=\hat{K}t^{\hat{m}}] = p, \text{ and, }  \textbf{Pr}[Y_i=0]=1-p, \text{ for } 1 \leq i \leq l.
\end{equation}

Now, we estimate the total number of possible hazard instances in the software, $P$. Similar to the case of random variable $X$, since all the modules are tested randomly and independently, we treat each pair of predictions as mutually independent. Furthermore, we assume that the assumptions \ref{Assumption1}-\ref{Assumption3}, and \ref{Assumption5}-\ref{Assumption9}, must hold true for defining the following random variable, $Y$. Now, in a software, the total number of possible hazard instances ($Y$) is estimated as follows: 

\begin{equation}
\label{SumofIndependentTrials-Y}
    Y = \sum_{i=1}^{l} Y_i
\end{equation}
Here, the random variable $Y$ represents the total possible hazard instances in the software, $P$. The random variable derived in Equation \ref{SumofIndependentTrials-Y} is nothing but the hazard rate in the SDP-based software, that satisfies Assumption \ref{Assumption10}. Now, let $\hat{z}_Y(t)$ be the hazard rate of the SDP-based software, then we have that:
\begin{equation}
\label{HazardRate-Y}
    \hat{z}_Y(t) = Y
\end{equation}

Similar to the definition of the random variable $X_i$, since the SDP model predicts the class label for each newly developed software module independently with similar probability $p$, $Y_i$ becomes a Bernoulli trial. Now, for the independent Bernoulli random variables $Y_1, Y_2, \cdots, Y_l$, their sum ($Y$) is in binomial distribution. Below, we derive the expected number of hazards that may occur in the entire software. The mean of the binomial distribution is derived as follows (by the linearity of expectation):
\begin{multline}
\label{Expectation-Y}
    \mu_{\hat{z}_Y} = \mathbb{E}[Y] = E\Big[\sum_{i=1}^{l} Y_i\Big] = \sum_{i=1}^{l} \mathbb{E}[Y_i] \\
    = \sum_{i=1}^{l} \big[[\hat{K}t^{\hat{m}}*\textbf{Pr}[\textit{Y}_i=\hat{K}t^{\hat{m}}]+0*\textbf{Pr}[\textit{Y}_i=0]\big] \\ = \sum_{i=1}^{l} p\hat{K}t^{\hat{m}} = lp\hat{K}t^{\hat{m}}
\end{multline}
The expected number of hazard instances (that is, $\mu_{\hat{z}_Y}$) in Equation \ref{Expectation-Y} are nothing but the scaled form of the Weibull distribution of the hazard instances in the software $P$. In this case, each random variable $Y_i$ assumes a value $\hat{K}t^{\hat{m}}$ with the probability $p$, and 0 with probability, $1-p$. In Equation \ref{Expectation-Y}, we follow the linearity of expectation because, the random variables $Y_i, i\in\{1,2, \cdots, l\}$ are assumed to be independent of each other. Now, in the following subsections, we derive the tight bounds based on the hazard rate (and reliability) of the software.
\subsubsection{The tight bounds in terms of the hazard rate of a software}
Now, using the facts from Equations \ref{SumofIndependentTrials-Y}, and \ref{Expectation-Y}, and from the Definition \ref{Definition: Weibull Hazard}, we provide tight bounds on the possibility of achieving fewer hazards in software that is implemented using SDP models than manually tested software.
\begin{theorem}
\label{Theorem-Weibull-Hazard-Y}
Let $Y_1, Y_2, \dots, Y_l$ be the independent Bernoulli trials such that for, $1\leq i \leq l$, $Pr[Y_i = \hat{K}t^{\hat{m}}] = p$, where $0 < p < 1$. Then $\exists$ parameters $K>0, m>-1, \hat{K}, \hat{m}$, time $t>0$. Then for Y = $\sum_{i=1}^{l} Y_i, \mu_{\hat{z}_Y} = \mathbb{E}[Y] = lp\hat{K}t^{\hat{m}}$, and for the Weibull hazard rate, $Kt^m$:
\begin{align}
    Pr[Y < Kt^m] < e^{\frac{-[lp\hat{K}t^{\hat{m}}-Kt^m]^2}{2lp\hat{K}t^{\hat{m}}}}
\end{align}
\end{theorem}
\begin{proof}
We know from Equation \ref{Expectation-Y}, the expected number of hazards (which are assumed to be observed from the dormant defective modules, and each dormant defect can cause some specific form of Weibull distribution of the hazard rate) in a software system that uses the SDP model in a given time period $[0,t]$:
\begin{align*}
    \mu_{\hat{z}_Y} = lp\hat{K}t^{\hat{m}}
\end{align*}
We know for some $0<\delta \leq 1$, and $\mu$, and from Equation \ref{LowerChernoff}, using Chernoff bound, the lower tail bound for the sum of independent Bernoulli trials, \textit{Y}, is:
\begin{align*}
    Pr[Y < (1-\delta)\mu] < e^{\frac{-\mu\delta^2}{2}}
\end{align*}
Now, we wish to obtain a tight lower bound that the random variable, $Y$, deviates far from the hazard rate of the manually tested software, $Kt^m$, in a given time period $[0,t]$. By equating the value of $(1-\delta)\mu$ in Equation \ref{LowerChernoff} with $z(t)$ in \ref{Weibull-Hazard}, then we have:
\begin{align*}
    (1-\delta)\mu = Kt^m
\end{align*}
Since, we are calculating the possible tight bound for the deviation of the random variable $Y$, we denote the expected number of hazards as $\mu_{\hat{z}_Y}$, in place of $\mu$. Now, substituting $\mu_{\hat{z}_Y}$ (from Equation \ref{Expectation-Y}) in the above equation, then we have that:
\begin{align}
\label{Weibull-Hazard-Delta-Y}
    \Rightarrow \delta = 1 - \frac{Kt^m}{lp\hat{K}t^{\hat{m}}}
\end{align}
Here, the value of $\delta$ defines the left margin value from the expectation, $\mu_{\hat{z}_Y}$. 

Since the hazard rate of the software system (which is tested manually) is assumed to follow the Weibull distribution, the tight lower bound on the probability of the hazard rate of the software system (which is tested using the SDP model), which is deviated below from the expected value $Kt^m$ (where $Kt^m$ is the scaled down value of the expectation, $\mu_{\hat{z}_Y}$), is derived using Equation \ref{LowerChernoff} as:
\begin{align}
    Pr[Y < Kt^m] < e^{\frac{-lp\hat{K}t^{\hat{m}}\big[1-\frac{Kt^m}{lp\hat{K}t^{\hat{m}}}\big]^2}{2}}
\end{align}
In the above equation, for some $K>0, m>-1, \hat{K}>0, \hat{m}>-1,$ and $t>0$, the value $Kt^m$ is assumed to be below the expectation, $lp\hat{K}t^{\hat{m}}$. Now, few steps of simplification results in the proof for this Theorem.
\end{proof}
The Theorem \ref{Theorem-Weibull-Hazard-Y} shows that the probability of occurrence of fewer hazards in the software that uses SDP is lower than the total hazards in the same software that is tested by humans, and is bound by the maximum value of $e^{\frac{-[lp\hat{K}t^{\hat{m}}-Kt^m]^2}{2lp\hat{K}t^{\hat{m}}}}$. Aforementioned, here, the proof is constructed based on each random variable $Y_i$ takes a value from the specific form of the Weibull distribution. Thus, the difference between the Theorems \ref{Theorem-Weibull-Hazard} and \ref{Theorem-Weibull-Hazard-Y} is based on the definition of the expectation of the random variable.
\subsubsection{The tight bounds in terms of the reliability of a software}
In this section, we first provide a lemma that calculates the reliability of a software from the total hazards in a software (from Equation \ref{SumofIndependentTrials-Y}, it is represented as $Y$, at any time, $t>0$) that is tested by using the SDP model. Later, we provide a tight upper bound for the deviation of a random variable (represented as the reliability of a software that is tested by using the SDP model) from the reliability of a manually tested software.

\begin{lemma}
\label{Lemma-Reliability-Y}
Let $Y_1, Y_2, \dots, Y_l$ be the independent Bernoulli trials, then for the hazards in the software system, Y = $\sum_{i=1}^{l} Y_i$, its reliability is:
\begin{align}
\label{Reliability-Y}
    \hat{R}_Y(t) = e^{-Yt}
\end{align}
\end{lemma}
\begin{proof}
We know that the hazards in software that are tested by using the SDP are random variable $Y$. Then, using Equation \ref{Reliability from Hazard Rate}:
\begin{align}
\label{Reliability-Y-Initial}
    \hat{R}_Y(t) = e^{-\int_0^t Y dx}
\end{align}
Here, we use $\hat{R}_Y(t)$ to represent the reliability of the software that is tested by using the SDP model, where each misclassified defective module generates a specific form of Weibull distribution of the hazard rates in the software. Now, simplifying Equation \ref{Reliability-Y-Initial} will results in the reliability of the SDP-based software.
\end{proof}
Similar to the hazards case, to find the tight bound for the deviation of a random variable (represented using the reliability of a software that is tested by using the SDP model) far from the reliability of a manually tested software (that is, $e^{\frac{-Kt^{m+1}}{m+1}}$), we require to compute the expected reliability of a software that is tested by using the SDP model. From Lemma \ref{Lemma-Reliability-Y}, we know the reliability of a software that is tested by using the SDP model. Now the expected reliability ($\mu_{\hat{R}_Y}$ or $\mathbb{E}[\hat{R}_Y(t)]$) is derived as:
\begin{equation}
\mathbb{E}[\hat{R}_Y(t)] = \mu_{\hat{R}_Y} = \mathbb{E}[e^{-Yt}] = \mathbb{E}\Big[e^{-t\sum_{i=1}^{l} Y_i}\Big]
\end{equation}
Since the random variables ($Y_i$) are assumed to be independent, then the sum of the terms in the exponent will become the product of the exponential terms. This is given as:
\begin{equation}
\label{Expected Reliability-SDP-Y-1}
\mathbb{E}\Big[e^{-t\sum_{i=1}^{l} Y_i}\Big] = \prod_{i=1}^{l} \mathbb{E}\big[e^{-tY_i}\big]
\end{equation}
Here, the random variable, $e^{-tY_i}$ assumes a value $e^{\hat{K}t^{\hat{m}+1}}$ with probability $p$, and the value 1 with probability $1-p$. Now, using these facts, we have that from Equation \ref{Expected Reliability-SDP-Y-1}:
\begin{equation}
\label{Expected Reliability-SDP-Y-2}
\prod_{i=1}^{l} \mathbb{E}\big[e^{-tY_i}\big] = \prod_{i=1}^{l} \big[pe^{\hat{K}t^{\hat{m}+1}}+1-p\big] = \prod_{i=1}^{l} \big[p(e^{\hat{K}t^{\hat{m}+1}}-1)+1\big]
\end{equation}
We know the fact that, $1+x < e^x$. Now, using this inequality with $x = p(e^{\hat{K}t^{\hat{m}+1}}-1)$, we rewrite Equation \ref{Expected Reliability-SDP-Y-2} to obtain the expected reliability:
\begin{multline}
\label{Expected Reliability-SDP-Y}
\mathbb{E}[\hat{R}_Y(t)] = \mu_{\hat{R}_Y} = \\ \prod_{i=1}^{l} \big[p(e^{\hat{K}t^{\hat{m}+1}}-1)+1\big] < \prod_{i=1}^{l} e^{p(e^{\hat{K}t^{\hat{m}+1}}-1)} \\= e^{lp(e^{\hat{K}t^{\hat{m}+1}}-1)}
\end{multline}
The above inequality does not harm the following final bound in Theorem \ref{Theorem-Weibull-Reliability-Y}. Now using the Lemmas \ref{Lemma-Reliability-Wibull} and \ref{Lemma-Reliability-Y}, below, we provide tight bounds on achieving the maximum reliability in software (that is implemented using SDP models) compared to manually tested software.
\begin{theorem}
\label{Theorem-Weibull-Reliability-Y}
Let $Y_1, Y_2, \dots, Y_l$ be the independent Bernoulli trials such that for, $1\leq i \leq l$, $Pr[Y_i = Kt^m] = p$, where $0 < p < 1$. Then $\exists$ parameters $K>0, m>-1, \hat{K}>0, \hat{m}>-1$, and time $t>0$. Then for Y = $\sum_{i=1}^{l} Y_i, \mu_{\hat{z}_Y} = \mathbb{E}[\hat{R}_Y(t)] < e^{lp(e^{\hat{K}t^{\hat{m}+1}}-1)}$, and for the Weibull Reliability function, $e^{\frac{-Kt^{m+1}}{m+1}}:$
\begin{align}
    Pr\Big[e^{-Yt} > e^{\frac{-Kt^{m+1}}{m+1}}\Big] < e^{\frac{-\Big[e^{lp(e^{\hat{K}t^{\hat{m}+1}}-1)}-\frac{Kt^m}{m+1}\Big]^2}{2e^{lp(e^{\hat{K}t^{\hat{m}+1}}-1)}}}
\end{align}
\end{theorem}
\begin{proof}
The proof for this upper tail is very similar to the proof for the lower tail, as we saw in Theorem \ref{Theorem-Weibull-Hazard-Y}. From the Lemmas \ref{Lemma-Reliability-Y} and \ref{Lemma-Reliability-X}, as before,
\begin{multline}
\label{Formatting-Y}
     Pr\Big[e^{-Yt} > e^{\frac{-Kt^{m+1}}{m+1}}\Big] =  Pr\Big[e^{Yt} < e^{\frac{Kt^{m+1}}{m+1}}\Big] \\
     = Pr\Big[Yt < \frac{Kt^{m+1}}{m+1}\Big]
     = Pr\Big[Y < \frac{Kt^m}{m+1}\Big]
\end{multline}
Initially, we wish to obtain a tight upper bound for a random variable, $e^{-Yt}$, deviates far from the value $e^{\frac{-Kt^{m+1}}{m+1}}$. Now, from Equation \ref{Formatting-Y}, we wish to obtain the tight lower bound for a random variable $Y$, deviates far below from the value, $\frac{Kt^m}{m+1}$. Here, from Equations \ref{LowerChernoff} and \ref{Formatting-Y}, for some $K>0, m>-1, \hat{K}>0, \hat{m}>-1,$ and $t>0$, the value $\frac{Kt^m}{m+1}$ is assumed to be below the expectation, $e^{lp(e^{\hat{K}t^{\hat{m}+1}}-1)}$. Now, equating the $(1-\delta)\mu$ in Equation \ref{LowerChernoff} with $\frac{Kt^m}{m+1}$ in \ref{Formatting-Y}, then we have:
\begin{align}
\label{Delta-Previous Form-Y}
    (1-\delta)\mu = \frac{Kt^m}{m+1}
\end{align}
Since, we are calculating the possible tight bound for the deviation of the random variable $Y$, we denote the expected number of hazards as $\mu_{\hat{z}_Y}$, in place of $\mu$. Now, substitute the value of $\mu_{\hat{z}_Y}$ (from Equation \ref{Expectation-Y}) in Equation \ref{Delta-Previous Form-Y}, then we have:
\begin{align}
\label{Weibull-Reliability-Delta-Y}
    \Rightarrow \delta = 1 - \frac{Kt^m}{(m+1)e^{lp(e^{\hat{K}t^{\hat{m}+1}}-1)}}
\end{align}
%Here, the value of $\delta$ defines the left margin value from the expectation, $\mu_{\hat{z}_Y}$.
Substitute the value of $\delta$ (from Equation \ref{Weibull-Reliability-Delta-Y}) in Equation \ref{LowerChernoff} then:
\begin{align}
    Pr\Big[Y < \frac{Kt^m}{m+1}\Big] < e^{\frac{e^{lp(e^{\hat{K}t^{\hat{m}+1}}-1)}\bigg[1 - \frac{Kt^m}{(m+1)e^{lp(e^{\hat{K}t^{\hat{m}+1}}-1)}}\bigg]^2}{2}}
\end{align}
A few steps of simplification will ensure the proof of this theorem.
\end{proof}
The Theorem \ref{Theorem-Weibull-Reliability-Y} provides evidence that the probability of obtaining better reliability in software that uses the SDP model is higher than the reliability of the same software that is tested by humans and is bound by the maximum value of $e^{\frac{-\Big[e^{lp(e^{\hat{K}t^{\hat{m}+1}}-1)}-\frac{Kt^m}{m+1}\Big]^2}{2e^{lp(e^{\hat{K}t^{\hat{m}+1}}-1)}}}$. Similar to the result of Theorem \ref{Theorem-Weibull-Hazard-Y}, the proof of Theorem \ref{Theorem-Weibull-Reliability-Y} is constructed based on each random variable $Y_i$ takes a value from the specific form of the Weibull distribution (later, using this hazard model, we calculated the reliability). Thus, the difference between the Theorems \ref{Theorem-Weibull-Reliability} and \ref{Theorem-Weibull-Reliability-Y} is based on the definition of the expectation of the random variable.  %In this case, as discussed above, each misclassified defective module generates a Weibull distribution of the hazard rate (later, we estimate the reliability from the hazard rate) with parameters $\hat{K}, \hat{m}$, and $t$. Thus the total hazard instances in a software that is tested by using the SDP model is a scaled form of the Weibull distribution.

The Theorems \ref{Theorem-Weibull-Hazard-Y} and \ref{Theorem-Weibull-Reliability-Y} are obtained by altering the definition of the random variable. From Theorems \ref{Theorem-Weibull-Hazard-Y} and \ref{Theorem-Weibull-Reliability-Y}, similar to the supplementary proofs provided in Section \ref{Supplements}, we can derive numerous proofs in terms of both the hazard rate and reliability.
\subsection{Analysis of the Theorems \ref{Theorem-Weibull-Hazard}, \ref{Theorem-Weibull-Reliability}, \ref{Theorem-Weibull-Hazard-Y}, and \ref{Theorem-Weibull-Reliability-Y}}
\label{Analysis of the Theorems 1 to 4}
%The Theorem \ref{Theorem-Weibull-Hazard} describes the sharp bounds only by considering the expected value of the random variable, $X$. Now, substituting the value of $\mu_{\hat{z}}$ (that is, $lp$) in Theorem \ref{Theorem-Weibull-Hazard}, then we have that:
%\begin{align*}
%    Pr[X < Kt^m] < e^{\frac{-(\mu_{\hat{z}}-Kt^m)^2}{2\mu_{\hat{z}}}} = e^{\frac{-(lp-Kt^m)^2}{2lp}}
%\end{align*}
The change in the behaviour of the bounds in Theorems \ref{Theorem-Weibull-Hazard} and \ref{Theorem-Weibull-Reliability} can be described based on (1) the number of predicted clean modules and (2) the time at which we observe the hazard values. First, since the probability ($p$) is a fixed value for the target application, assume that at any arbitrary fixed values of the parameters such as $K, m$ (assume $K$ and $m$ are derived based on some criteria)\footnote{Estimating the parameters such as $K, m$, at any time $t$, are important in constructing the proofs. However, for this preliminary work, we are assuming that the values of these parameters must obey the condition imposed by $\delta$ (that is, $0<\delta\leq1$).}, and time $t$, the probability of observing the minimum hazards in a software that is tested by using the SDP models when compared with the hazard rate of a manually tested software is exponentially small in $l$, implying that the larger the predicted clean modules, the more likely it is that the true colours show through. Second, assume the parameters $K$ and $m$ are derived based on some criteria in a manually tested software. Now, as the time progresses (that is, for the large values of time, $t$), for some large values of the parameters $K$ and $m$ (assume these parameters must obey in accordance with the condition imposed by $\delta$ (that is, $0<\delta\leq1$)), the terms $e^{\frac{-(lp-Kt^m)^2}{2lp}}$ (from Theorem \ref{Theorem-Weibull-Hazard}) and $e^{\frac{-e^{lp(1-e^{-t})}}{2}\big[e^{lp(e^{-t}-1)}-\frac{Kt^m}{m+1}\big]^2}$ (from Theorem \ref{Theorem-Weibull-Reliability}) will become extremely small or even close to 0. Hence, as the time progresses (assume we already have observed the $l$ number of clean modules), we can derive an ineffectiveness of the SDP models (represented in terms of hazard rate and reliability) when compared with the manually tested software systems, in the real-world testing environments.

Since the Corollaries presented in Section \ref{Supplements} are derived from the Theorems \ref{Theorem-Weibull-Hazard} and \ref{Theorem-Weibull-Reliability}, the above same analysis is applicable to describe the behaviour of the bounds.

Similar to the analysis provided for the Theorems \ref{Theorem-Weibull-Hazard} and \ref{Theorem-Weibull-Reliability}, the change in the behaviour of the bounds in Theorems \ref{Theorem-Weibull-Hazard-Y} and \ref{Theorem-Weibull-Reliability-Y} also can be described based on, (1) the number of predicted clean modules, (2) the time at which we observe the hazard values. First, similar to the above analysis, since the probability ($p$) is a fixed value for the target application, assume at any fixed values of the parameters - $K, m, \hat{K}, \hat{m}$\footnote{Estimating the parameters such as $K, m,\hat{K}$, and $\hat{m}$ at any time $t$, are important in constructing the proofs. However, for this preliminary work, we are assuming that the values of these parameters must obey the condition imposed by $\delta$ (that is, $0<\delta\leq1$). Among all these parameters, $\hat{K}$ and $\hat{m}$ are used to describe the hazard rate induced from the misclassified defective module. We assume that a similar form of the Weibull distribution (that is, $\hat{K}t^{\hat{m}}$) is used to describe the hazard incidents in a software that uses SDP model. Since the Weibull distribution is an ideal representation of the hazard rate of a software or a software component (module), we apply the same representation to describe the hazard rate of each misclassified defective module.}, and time $t$, when compared with the hazard rate of a manually tested software, the probability of observing the minimum hazards in a software which is tested by using the SDP models is exponentially small in $l$, or even close to zero, implying that, with a large number of predicted clean modules, we are more likely to observe the true colours of the developed prediction model. Second, assume the parameters $K$ and $m$ are derived based on some criteria in a manually tested software. Now, as the time progresses (that is, for the large values of time, $t$), for some large values of the parameters $K, m, \hat{K}$, and $\hat{m}$ (assume these parameters must obey in accordance with the condition imposed by $\delta$ (that is, $0<\delta\leq1$)), the terms $e^{\frac{-[lp\hat{K}t^{\hat{m}}-Kt^m]^2}{2lp\hat{K}t^{\hat{m}}}}$ (from Theorem \ref{Theorem-Weibull-Hazard-Y}) and $e^{\frac{-\Big[e^{lp(e^{\hat{K}t^{\hat{m}+1}}-1)}-\frac{Kt^m}{m+1}\Big]^2}{2e^{lp(e^{\hat{K}t^{\hat{m}+1}}-1)}}}$ (from Theorem \ref{Theorem-Weibull-Reliability-Y}) will become extremely small or even close to 0. As a result, as time passes (assuming we have already observed the $l$ number of clean modules), we can conclude that the SDP models are ineffective in real-world testing environments.

In the above analysis, we are using a specific form of the hazard rate to describe how each defective module causes system failures. We chose the special form of the Weibull distribution to just describe that the random variable may take any form of the failure function instead of simply taking either 1 or 0. It is worth noting that the parameters such as $\hat{K}$, and $\hat{m}$ can be estimated prior to providing the proofs (in Theorems \ref{Theorem-Weibull-Hazard-Y} and \ref{Theorem-Weibull-Reliability-Y}).

\section{Conclusion}
\label{Conclusion}
This work aims to fill a gap in the literature by addressing the feasibility of the SDP models in real-world testing environments. We made a few assumptions in order to provide theoretical proofs for the feasibility of the binary defect prediction models. These assumptions enable easy computations in providing the proofs. Relaxing the stated assumptions may lead us to observe sharp bounds.
%This is because, at the time of testing the software, the tester observes two types of predictions (such as defective or clean) and, the predicted clean modules may contain defective instances. 

To provide the feasibility of the SDP models, first, from the test set, we measured the probability value as the percentage of the false negatives over the total observed cleans. Assuming that any optimal classifier produces both the \textit{defect injections} (also called the false negatives) and \textit{original cleans} (also called true negatives) on the target project, and each defective injection can cause a single failure in the system, using this probability, we defined the distribution of the failures in the target software project. We then calculated the expected number of failure incidents in the software as a result of the defective injections.

Using the hazard rate and its reliability of the same software that is tested under the manual testing conditions, we provided the tight bounds using the famous tail inequality technique called the Chernoff bound. Our approach is mainly focus on providing the bounds by comparing the hazard rate (and its reliability) of the software that is tested using the SDP model with the hazard rate (and its reliability) of the manually tested software. The tight bounds are provided by assuming the manually tested software follows the Weibull distribution of the hazard rate. These bounds are dynamic and can be replaced with any hazard function at hand to measure the feasibility of the SDP models. To illustrate the possibility of using various hazard models, we have also provided some proofs based on the supplements of the Weibull distribution of the hazard models.

In terms of the Weibull hazard rate, we derived a maximum bound of $e^{\frac{-(\mu_{\hat{z}}-Kt^m)^2}{2\mu_{\hat{z}}}}$ as the probability of observing fewer hazards in the software (that uses SDP) than the total hazards in the same software that is tested by humans. Similarly, in terms of the reliability of a software that is derived from the Weibull distribution of the hazard rate, we derived a maximum bound of $e^{-\big[e^{lp(e^{-t}-1)}-\frac{Kt^m}{m+1}\big]^2\frac{1}{2e^{lp(e^{-t}-1)}}}$ as the probability of obtaining more reliability in the software that uses SDP than the reliability of the same software that is tested by humans. The summary of these bounds indicates that, as the random variable (that is modelled as the hazard rate or the reliability) deviates far from the expectation, we conclude the inefficiency of the defect prediction models in the real-world testing environment.

In the above bounds, we assumed the occurrence of a single failure from each misclassified defective instance. In addition to the above bounds, we also derived two more bounds in terms of the hazard rate and reliability by assuming each misclassified defective instance generates a Weibull distribution of the hazard rate instead of generating a single failure in the system. In terms of the Weibull hazard rate, we derived a maximum bound of $e^{\frac{-[lp\hat{K}t^{\hat{m}}-Kt^m]^2}{2lp\hat{K}t^{\hat{m}}}}$ as the probability of observing fewer hazards in the software (that uses SDP) than the total hazards in the same software that is tested by humans. Similarly, in terms of the reliability of a software that is derived from the Weibull distribution of the hazard rate, we derived a maximum bound of $e^{\frac{-\Big[e^{lp(e^{\hat{K}t^{\hat{m}+1}}-1)}-\frac{Kt^m}{m+1}\Big]^2}{2e^{lp(e^{\hat{K}t^{\hat{m}+1}}-1)}}}$ as the probability of obtaining more reliability in the software that uses SDP than the reliability of the same software that is tested by humans. From these two bounds, we can derive numerous proofs by simply substituting the Weibull hazard function with any other function form of the hazard rate.

We believe that providing a critique of the developed binary classification model (in our case, the software defect prediction model) in a real-world testing environment is novel in machine learning research and has the potential to provide insight into the feasibility of other applications. Some of the potential future works based on this idea are listed in the following section.
\subsection{Future Plans}
\label{Future Plans}
The theorems \ref{Theorem-Weibull-Hazard}, \ref{Theorem-Weibull-Reliability}, \ref{Theorem-Weibull-Hazard-Y}, and \ref{Theorem-Weibull-Reliability-Y} provide preliminary bounds on the feasibility of the developed binary classification model (SDP model) in real-world testing environments. Within the scope of this work, the extensions of Theorems \ref{Theorem-Weibull-Hazard}, \ref{Theorem-Weibull-Reliability}, \ref{Theorem-Weibull-Hazard-Y}, and \ref{Theorem-Weibull-Reliability-Y} are numerous. A few examples include:
\begin{enumerate}
    \item In this work, we assume a Weibull model of the hazard rate for the manually tested software in providing the proofs. The bounds provided in theorems \ref{Theorem-Weibull-Hazard}, \ref{Theorem-Weibull-Reliability}, \ref{Theorem-Weibull-Hazard-Y}, and \ref{Theorem-Weibull-Reliability-Y} become more application-specific if the state-of-the-art hazard (and reliability) models are used in the construction of the proof instead of using the Weibull distribution.
    \item This work assumes independence among the random variables (among $X_i$s as well as, among $Y_i$s). Instead, new bounds will be derived assuming the dependency among the random variables. This ensures deriving the bounds in the presence of cascading failures in the software system. Because it is a fact that the coupling between the modules is present in any software \cite{pressman2005software}, the occurrence of failure in one module may cause failure in the coupled module \cite{lyu1996handbook}. As a result, the software may experience more failures due to the cascading effect \cite{lyu1996handbook}. In this scenario, the dependence among the random variables ensures the coupling between the software modules. Hence, finding the usefulness of the defect prediction model in the real-world testing environments assuming the dependency among the random variables is a well needed research.
\end{enumerate}

We hope this work is the basis for some potential future research directions. That is, the generalizability of this work. We hope that this work opens up a new avenue for analysing the impact of a prediction model (especially a two-class classification model) on the targeted application (of any kind) using probabilistic bounds. A few example applications include: optical character recognition (OCR) systems, machine translation systems, fraud detection systems, tumour detection systems, etc.
\bibliographystyle{elsarticle-num} 		
\bibliography{PostSDP}

%\appendix
\end{document}